\documentclass[onecolumn,10pt]{IEEEtran}
\usepackage{graphicx}
\usepackage{amsmath,amssymb}
\usepackage{cases}
\usepackage{color}
\usepackage{amsfonts}
\usepackage{indentfirst}
\usepackage{slashbox}
\usepackage{float}
\usepackage{cite}
\usepackage{caption}
\usepackage{subcaption}
\usepackage{algorithm}
\usepackage{algpseudocode}
\usepackage{booktabs}
\usepackage{multirow}
\usepackage{amsthm}
\usepackage{epstopdf}
\usepackage{diagbox}
\usepackage{blkarray}
\makeatletter
\newtheoremstyle{mythm}{3pt}{3pt}{}{16pt}{\bfseries}{:}{.5em}{}
\renewcommand{\arraystretch}{1.2}
\theoremstyle{mythm}
\newtheorem{theorem}{Theorem}
\newtheorem{example}{Example}
\newtheorem{definition}{Definition}
\newtheorem{remark}{Remark}

\newtheorem{lemma}{Lemma}
\newtheorem{construction}{Construction}

\begin{document}
\title{Coded Caching for Combination Networks with Multiaccess
\author{Leitang Huang, Jinyu Wang and Minquan Cheng}
\thanks{Huang, Wang and Cheng are with Guang xi Key Lab of Multi-source Information Mining $\&$ Security, Guangxi Normal University,
Guilin 541004, China (e-mail: $\{$hlt11153,jjiang2008,chengqinshi$\}$@hotmail.com). Wang is also with College of Mathematics and Statistics, Guangxi Normal University, Guilin 541004, China.}
}

\date{}
\maketitle
\begin{abstract}
In a traditional $(H, r)$ combination network, each user is connected to a unique set of $r$ relays. However, few research efforts to consider $(H, r, u)$ multiaccess combination network problem where each $u$ users are connected to a unique set of $r$ relays. A naive strategy to obtain a coded caching scheme for $(H, r, u)$ multiaccess combination network is by $u$ times repeated application of a coded caching scheme for a traditional $(H, r)$ combination network. Obviously, the transmission load for each relay of this trivial scheme is exactly $u$ times that of the original scheme, which implies that as the number of users multiplies, the transmission load for each relay will also multiply. Therefore, it is very meaningful to design a coded caching scheme for $(H, r, u)$ multiaccess combination network with lower transmission load for each relay.
In this paper, by directly applying the well known coding method (proposed by Zewail and Yener) for $(H, r)$ combination network, a coded caching scheme (ZY scheme) for $(H, r, u)$ multiaccess combination network is obtained. However, the subpacketization of this scheme has exponential order with the number of users, which leads to a high implementation complexity. In order to reduce the subpacketization, a direct construction of a coded caching scheme for $(H, r, u)$ multiaccess combination network is proposed by means of Combinational Design Theory, where the parameter $u$ must be a combinatorial number. For arbitrary parameter $u$, a hybrid construction of a coded caching scheme for $(H, r, u)$ multiaccess combination network is proposed based on our direct construction. Theoretical and numerical analysis show that our last two schemes have smaller transmission load for each relay compared with the trivial scheme, and have much lower subpacketization compared with ZY scheme.
 \end{abstract}

\begin{IEEEkeywords}
Coded caching, Placement delivery array, Combination networks, Multiaccess.
\end{IEEEkeywords}
\section{Introduction}
\label{introduction}

Coded caching proposed by Maddah-Ali and Niesen(MN) in \cite{MN} not
only utilizes the users' cache memories in shifting some of the network traffic to off-peak hours, but also creates multicast
opportunities that further reduce network congestion during peak traffic hours. The first coded caching system studied is the following shared-link broadcast network with end-user-caches: there exists a single sever with $N$ files with the same length connecting to $K$ users through a shared error-free broadcast link, where each user has a cache of size $M$ files. An $F$-division $(K,M,N)$ coded caching scheme consists of two phases: placement phase during the off-peak traffic times and delivery phase during peak traffic times. In the placement phase, the server divides each file into $F$ packets ($F$ is referred to as the subpacketization) with the same length, and then places some packets of size $M$ files in each user's cache without knowledge of later users' demands. If the server directly places some packets in each user's cache, without coding, the placement is said to be uncoded. Otherwise, the placement is called coded.
In the delivery phase, each user's demand arrives at the server. According to the users' demands, the server transmits coded packets (XOR of the required packets) in order to satisfy all users' demands with the help of their caches. The maximum normalized transmission amount over all possible demands is defined as the transmission load $R$, and the coded caching gain is defined as $\frac{K(1-\frac{M}{N})}{R}$, where $K(1-\frac{M}{N})$ is the transmission load of the conventional uncoded caching scheme.
The goal is to design a coded caching scheme with transmission load $R$ and the subpacketization $F$ as small as possible due to the requirements of efficiency of transmission and low implementation complexity. The above coded caching model is called a $(K,M,N)$ caching or MN caching system, for which Maddah-Ali and Niesen \cite{MN} proposed the first coded caching scheme (called MN scheme), which is optimal under the constraint of uncoded placement and $K\leq N$ \cite{WTP2020}. Yan et al. \cite{YCTC} proposed a combinatorial structure called placement delivery array (PDA) to design a coded caching scheme for MN caching system. It is worth noting that MN scheme can be represented by a special PDA which is called the MN PDA.

\subsection{Traditional combination network caching system}
\label{sub-traditional-caching}
Compared with the MN caching system, a more practical system, where users may communicate with the server through intermediate relays, has gained attention. Since the analysis of relay networks with arbitrary topologies is challenging, many studies focus on a symmetric version of this general problem known as combination networks with end-user-caches. The $(H,r,M,N)$ combination network caching system was first proposed in \cite{JWTL} as follows:
a server with $N$ files connects $H$ relays (without caches), which in turns connect $K={H\choose r}$ users, where each user is equipped with a cache of size $M$ files and each user is connected to a unique set of $r$ relays. All links are assumed to be error-free and interference-free. The transmission load for each relay is defined as the maximum normalized transmission amount from the server to each relay over all relays and all possible demands. The objective is to design a coded caching scheme with the transmission load for each relay and the subpacketization as small as possible.

The first two schemes for combination networks were proposed in \cite{JWTL}, one based on uncoded placement and routing in the delivery phase, and the other based on the placement strategy of MN scheme and a linear code for delivery. The authors in \cite{ZY2} proposed a scheme based on coded placement that effectively splits the combination network into $H$ parallel MN caching systems, each of which serves ${H-1\choose r-1}$ users by using MN scheme.
This scheme works for any parameters $H$, $r$ and memory ratio $\frac{M}{N}$ and the achieved transmission load for each relay is approximately $\frac{1}{r}$ of the transmission load of MN scheme. Furthermore, it can be used to deal with any relay networks with arbitrary topologies.
There are some other studies on improving the transmission load for each relay based on MN scheme, for instance, \cite{WJPT1,WJPT3,WJPT4,WTPJ,WTJP,EMNG} etc. Since the subpacketization of MN scheme increases exponentially with the number of users, all of the above known schemes have high subpacketization, which implies that these schemes not only have high implementation complexity but also can not be used when the files in server are not large enough.

In order to design a scheme with low subpacketization for the combination network caching system, the authors in \cite{YWY} proposed the concept of combinatorial placement delivery array (CPDA), which is an expansion of PDA. When $r$ divides $H$ (denoted by $r|H$), based on MN scheme, they proposed a scheme with lower subpacketization while the achieved transmission load for each relay is equal to that of the scheme in \cite{ZY2}. However, its subpacketization still increases exponentially with the number of users. The authors in \cite{CLZW} showed that the strongly coloring PDA in \cite{YTCC} is a CPDA, which is constructed by the strong coloring in bipartite graph \cite{JA}. Consequently, two classes of CPDAs were obtained in \cite{CLZW}, which lead to schemes with lower subpacketization and without the limitation of $r|H$ compared with the scheme in \cite{YWY}.

\subsection{Multiaccess combination network caching system}
\label{sub-multiaccess-caching}
It should be noted that all of the above studies are under the assumption that each user is connected to a distinct set consisting of $r$ relays. However, the number of users in a network system may be time-varying\cite{ZZCC} in practice. So it always happens that multiple users are connected to the same $r$ relays when the number of users is larger than ${H\choose r}$. Therefore, it is meaningful to study a combination network caching system where multiple users are connected to the same $r$ relays. Assume that there are exactly $u$ users connected to a unique set of $r$ relays. Such setting is denoted by an $(H,r,u,M,N)$ multiaccess combination network caching system.

Clearly, an $(H,r,1,M,N)$ multiaccess combination network caching system is an $(H,r,M,N)$ combination network caching system proposed in \cite{JWTL}. So we can generate an $(H,r,u,M,N)$ multiaccess combination network caching scheme by repeated application of an $(H,r,M,N)$ combination network caching scheme $u$ times. However, in this naive strategy not only the coded caching gain among the users who are connected to the same $r$ relays will be omitted, but also the coded caching gain among the users who are connected to different relay subsets containing at least one common relay
will not be fully considered.
In order to make the transmission load for each relay as small as possible, we prefer to design a scheme with the coded caching gain as large as possible. This motivates us to study the multiaccess combination network caching problem, where the basic idea is to generate coded caching gain among all the users as large as possible instead of generating coded caching gain only among ${H\choose r}$ users who are connected to distinct sets of $r$ relays. Specifically, the following main results are obtained in this paper.

\begin{enumerate}
\item By using the idea proposed in \cite{ZY2}, a scheme (called ZY scheme) for the multiaccess combination network is obtained, which works for any parameters $H$, $r$, $u$ and any memory ratio $\frac{M}{N}$, and the achieved transmission load for each relay is approximately $\frac{1}{r}$ of the transmission load of MN scheme, since the coded caching gain among the users who are connected to the same relay is fully considered.  However, its subpacketizaiton increases exponentially with the number of users.
\item In order to reduce the subpacketization, we omit the coded caching gain among the users who are connected to the same $r$ relays and consider the coded caching gain among the users who are connected to different relay subsets containing at least one common relay as much as possible, and directly construct a CPDA $\mathbf{P}$ from the view point of Combinatorial Design Theory, which leads to a new scheme (called Scheme A) for the multiaccess combination network with lower subpacketization.
    Since some special combinatorial structure is used in this scheme, the parameter $u$ has to be a combinatorial number.
\item For arbitrary parameter $u$, a hybrid construction of CPDA is proposed based on the MN PDA and the above directly constructed CPDA $\mathbf{P}$, which leads to a new scheme (called Scheme B) for the multiaccess combination network. It is worth noting that Scheme B not only has the same coded caching gain among the users who are connected to different relay subsets containing at least one common relay as Scheme A, but also fully considers the coded caching gain among the users who are connected to the same $r$ relays by increasing some subpacketization.
\end{enumerate}

 The rest of this paper is organized as follows. In Section \ref{sec-pre}, the traditional combination network caching system is extended to the multiaccess combination network caching system, and the concepts of PDA and CPDA are briefly reviewed. In Section \ref{sec-first}, three schemes for the multiaccess combination network are proposed. In section \ref{sect-performance}, performance analysis is provided. Finally, conclusion is drawn in Section \ref{sec-conclusion} while some proofs are provided in the Appendices.

\section{preliminaries}
\label{sec-pre}
In this section, the multiaccess combination network caching system and the concepts of PDA and CPDA are introduced. First, the following notations are useful in this paper.
\begin{itemize}
  \item $[a:b]:=\{a,a+1,\ldots,b\}$, $[a]:=\{1,2,\ldots,a\}$.
  \item $mod(a,q)$ denotes the least non-negative residue of $a$ modulo $q$. $<a>_q:=mod(a,q)$ if $mod(a,q)\neq0$, otherwise $<a>_q:=q$.
  \item For any set $\mathcal{H}$ and for any positive integer $r$ with $r<|\mathcal{H}|$, ${\mathcal{H}\choose r}:=\{\mathcal{A}|\mathcal{A}\subseteq \mathcal{H},|\mathcal{A}|=r\}$, i.e., ${\mathcal{H}\choose r}$ is the collection of $r$-sized subsets of $\mathcal{H}$, where $|\cdot|$ denotes the cardinality of a set.
  \item Given an array $\mathbf{P}=(p_{j,k})_{j\in[F], k\in[K]}$ with alphabet $[S]\cup \{*\}$, we define $\mathbf{P}+a=(p_{j,k}+a)_{j\in[F], k\in[K]}$ and $\mathbf{P}\times a=(p_{j,k}\times a)_{j\in[F], k\in[K]}$  for any integer $a$, where $a+*=*, a\times *=*$.
  \item For any two vectors ${\bf x}$ and ${\bf y}$ with the same length, $d({\bf x},{\bf y})$ is the number of coordinates in which ${\bf x}$ and ${\bf y}$ differ, $wt({\bf x})$ is the weight of ${\bf x}$, i.e., the number of nonzero coordinates of ${\bf x}$. For example, if ${\bf x}=(0,1,2,1)$ and ${\bf y}=(0,1,3,4)$, then $d({\bf x},{\bf y})=2$ and $wt({\bf x})=3$.
  \item For any vector ${\bf f}$ with length $H$ and for any nonempty subset $\mathcal{T}\subseteq[H]$, ${\bf f}|_{\mathcal{T}}$ is a vector with length $|\mathcal{T}|$ obtained by taking only the coordinates with subscript $j \in \mathcal{T}$. For example, if ${\bf f}=(0,2,3,1)$ and $\mathcal{T}=\{1,3\}$, then ${\bf f}|_{\mathcal{T}}=(0,3)$.

\end{itemize}

\subsection{Placement Delivery Array}
Yan {\em et al.} \cite{YCTC} proposed a combinatorial structure, called placement delivery array (PDA), to characterize the placement phase and delivery phase of a scheme for the MN caching system simultaneously.
\begin{definition}(PDA,\cite{YCTC})
\label{def-PDA}
For positive integers $K$, $F$, $Z$ and $S$, an $F\times K$ array $\mathbf{P}=(p_{j,k})$, $j\in [F], k\in[K]$, composed of a specific symbol $"*"$ called star and $S$ integers in $[S]$, is called a $(K,F,Z,S)$ placement delivery array (PDA) if it satisfies the following conditions:
\begin{enumerate}
\item [C$1$.] Each column has exactly $Z$ stars.
\item [C$2$.] Each integer in $[S]$ occurs at lease once.
\item [C$3$.] For any two distinct entries $p_{j_1,k_1}$ and $p_{j_2,k_2}$, $p_{j_1,k_1}=p_{j_2,k_2}=s$ is an integer only if
  \begin{enumerate}
     \item [a.] $j_1\ne j_2$, $k_1\ne k_2$, i.e., they lie in distinct rows and distinct columns; and
     \item [b.] $p_{j_1,k_2}=p_{j_2,k_1}=*$, i.e., the corresponding $2\times 2$  subarray formed by rows $j_1,j_2$ and columns $k_1,k_2$ must be of the following form
  \begin{eqnarray}
  \label{form}
    \left(\begin{array}{cc}
      s & *\\
      * & s
    \end{array}\right)~\textrm{or}~
    \left(\begin{array}{cc}
      * & s\\
      s & *
    \end{array}\right).
  \end{eqnarray}
   \end{enumerate}
  \end{enumerate}
  \hfill $\square$
\end{definition}
\begin{example}\rm
\label{E-pda}
It is easy to verify that the following array is a $(6,4,2,4)$ PDA,
\begin{eqnarray*}
\mathbf{P}=\left(\begin{array}{cccccc}
*&*&*&1&2&3\\
*&1&2&*&*&4\\
1&*&3&*&4&*\\
2&3&*&4&*&*
\end{array}\right).
\end{eqnarray*}
\hfill $\square$
\end{example}
The first PDA was proposed by Maddah-Ali and Niesen in \cite{MN,YCTC}.
\begin{lemma}(MN PDA, MN scheme\cite{MN,YCTC})
\label{le-MN-PDA}
For any positive integers $K$, $M$ and $N$ with $M<N$, if $KM/N$ is an integer, there exists a $(K,{K\choose KM/N},{K-1\choose KM/N-1},{K\choose KM/N+1})$ PDA, which leads to a $(K,M,N)$ coded caching scheme (MN scheme) with subpacketization $F={K\choose KM/N}$ and transmission load $R=\frac{K(1-M/N)}{KM/N+1}$ for the MN caching system.
\hfill $\square$
\end{lemma}

\subsection{Multiaccess combination network caching system}
We consider an $(H,r,u,M,N)$ multiaccess combination network caching system (see Fig. \ref{model}) in which a server $S_{er}$ containing $N$ files $\mathcal{W}=\{W_1,W_2,\ldots,W_{N}\}$, each of which is uniformly distributed in $\{0,1\}^B$ for some positive integer $B$, connects $H$ relays through $H$ error-free and interference-free links. Each set of $r$ relays connects $u$ users,  so the total number of users is $K=u{H \choose r}$. Each user has a storage capacity of size $M$ files where $0<M<N$ and all relays have no cache memory. Each relay $h$ could broadcast the intermediate signals from the server to the users who is connected to it. The users is denoted by $\mathcal{K}=\{\mathbf{k}=(\mathcal{A},i)\ |\  \mathcal{A}\in {[H]\choose r}, i\in[u]\}$ and user $(\mathcal{A},i)$ is connected to relay $h$ if and only if $h\in \mathcal{A}$. The set of relays which connect user $\mathbf{k}\in \mathcal{K}$ is denoted by $\mathcal{A}_{\mathbf{k}}$.
\begin{figure}[!htbp]
  \centering
  \includegraphics[scale=0.4]{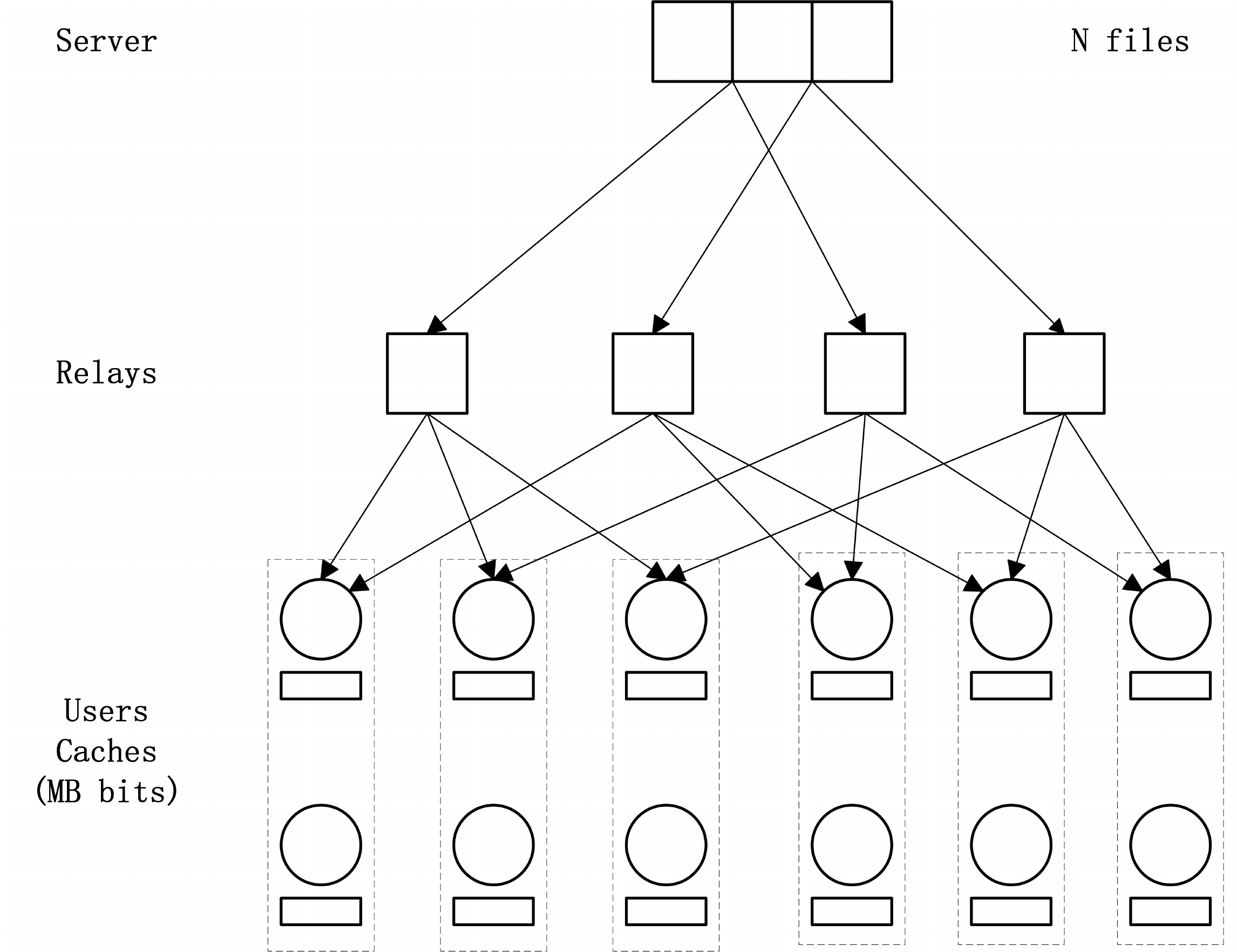}\\
  \caption{A multiaccess combination network with $H=4$, $r=2$, $u=2$}\label{model}
\end{figure}

An $F$-division $(H,r,u,M,N)$ multiaccess combination network caching scheme contains the following two phases.
\begin{itemize}
  \item {\bf Placement phase:} Each file is divided into $F$ packets with equal size. Then each user $\mathbf{k}\in\mathcal{K}$ directly accesses to the file library $\mathcal{W}$ and stores some packets or linear combinations of some packets of the $N$ files in its cache. The set of cached packets by user $\mathbf{k}$ is denoted by $\mathcal{Z}_{\mathbf{k}}$ whose size is at most $M$ files. That is, for each user $\mathbf{k}\in\mathcal{K}$, there exists a function $\phi_{\mathbf{k}}:\{0,1\}^{NB}\rightarrow \{0,1\}^{MB}$ to generate the cache contents $\mathcal{Z}_{\mathbf{k}}=\phi_{\mathbf{k}}((W_n)_{n\in[N]})$. Let $\mathcal{Z}=\{\mathcal{Z}_{\mathbf{k}}\ |\ \mathbf{k}\in\mathcal{K}\}$.
  \item {\bf Delivery phase:} Assume that each user $\mathbf{k}=(\mathcal{A},i)\in\mathcal{K}$ randomly requests one file from the file library $\mathcal{W}$. The demand vector is represented by $\mathbf{d}=(d_{\mathbf{k}})_{\mathbf{k}\in\mathcal{K}}$, which implies that user $\mathbf{k}$ requests the $d_{\mathbf{k}}^{\text{th}}$ file $W_{d_{\mathbf{k}}}$ where $d_{\mathbf{k}}\in[N]$. Given $(\mathcal{Z},\mathbf{d})$, the server sends a message $X_{S_{er}\rightarrow h}$ of size $L_{S_{er}\rightarrow h}$ bits to relay $h\in[H]$. Then relay $h\in[H]$ forwards $X_{S_{er}\rightarrow h}$ to its connecting users. User $\mathbf{k}$ can recover its desired file $W_{d_{\mathbf{k}}}$ by $\{X_{S_{er}\rightarrow h}\ |\ h\in \mathcal{A}\}$ with help of $\mathcal{Z}_{\mathbf{k}}$. This phase can be represented by the following encoding functions and decoding functions.
      \begin{itemize}
\item The $H$ encoding functions: For each relay $h\in[H]$,
$$\psi_{S_{er}\rightarrow h}\ :\  \{0,1\}^{NB}\times \{0,1\}^{KMB}\times[N]^K\rightarrow \{0,1\}^{L_{S_{er}\rightarrow h}}$$
generates the transmitted message $X_{S_{er}\rightarrow h}\triangleq \psi_{S_{er}\rightarrow h}(\mathcal{W},\mathcal{Z},{\bf d})$ from the server to the relay $h$. It is a function of the library $\mathcal{W}$, the cached contents of all users $\mathcal{Z}$ and the demand vector ${\bf d}$.

\item The $K$ decoding functions: For each user $\mathbf{k}=(\mathcal{A},i)\in\mathcal{K}$, $$\mu_{\mathbf{k}}\ :\ \{0,1\}^{\sum_{h\in \mathcal{A}}L_{S_{er}\rightarrow h}}\times \{0,1\}^{MB}\times [N]^K\rightarrow \{0,1\}^{B}$$ decodes the requested file of user $\mathbf{k}$ from all messages received by user $\mathbf{k}$ and its own cache, i.e.,
    $$W_{d_{\mathbf{k}}}=\mu_{\mathbf{k}}\left(\left\{X_{S_{er}\rightarrow h}\ |\ h\in \mathcal{A}\right\}, \mathcal{Z}_{\mathbf{k}},{\bf d}\right).$$
\end{itemize}

\end{itemize}

The transmission load for each relay is defined as
\begin{align*}
R=\max_{{\bf d}\in[N]^K,h\in[H]}\left\{\  \frac{L_{S_{er}\rightarrow h}}{B}\ \right\}.
\end{align*}

\subsection{Combinatorial placement delivery array}
\begin{definition}(CPDA,\cite{YWY})
\label{def-CPDA}
For any positive integers $H$, $r$ and $u$ with $r<H$, a $(K=u{H\choose r},F,Z,S)$ PDA $\mathbf{P}$ is called a $(K,F,Z,S)$ combinatorial placement delivery array (CPDA) if it satisfies the following condition.
\begin{enumerate}
\item [C$4$.] All the columns can be labeled by $\mathcal{K}=\{(\mathcal{A},i)\ |\  \mathcal{A}\in {[H]\choose r}, i\in[u]\}$ such that for any $s\in[S]$, the intersection of the first coordinate $\mathcal{A}$ of all column labels $(\mathcal{A},i)$ satisfying that $s$ appears in column $(\mathcal{A},i)$, denoted by $\mathcal{I}_s$, is not empty.
   \end{enumerate}
\hfill $\square$
\end{definition}

The definition of CPDA is proposed based on the definition of PDA. So a CPDA must be a PDA. However, the converse does not hold necessarily, since CPDA has another requirement, i.e., the condition C$4$. The detailed discussion of the relationship between PDA and CPDA could be found in \cite{CLZW}.
Based on a CPDA, a scheme for the multiaccess combination network can be obtained by using Algorithm \ref{alg:CPDA}.
\begin{algorithm}[htb]
\caption{Caching scheme based on $(K,F_1,Z,S)$ CPDA in \cite{CLZW}}\label{alg:CPDA}
\begin{algorithmic}[1]
\Procedure {Placement}{$\mathbf{P}$, $\mathcal{W}$}
\State Split each file $W_n\in\mathcal{W}$ into $F_1$ packets, i.e., $W_{n}=\{W_{n,j}\ |\ j\in[F_1]\}$.
\For{$\mathbf{k}\in \mathcal{K}$}
\State $\mathcal{Z}_{\mathbf{k}}\leftarrow\{W_{n,j}\ |\ p_{j,\mathbf{k}}=*, \forall~j\in[F_1], n\in [N]\}$
\EndFor
\EndProcedure
\Procedure{Delivery}{$\mathbf{P}, \mathcal{W},{\bf d}$}
\For{$s\in[S]$}
\State $X_s=\oplus_{p_{j,\mathbf{k}}=s,j\in[F_1],\mathbf{k}\in\mathcal{K}}W_{d_{\mathbf{k}},j}$;
\State $\mathcal{I}_s=\cap_{p_{j,\mathbf{k}}=s,j\in[F_1],\mathbf{k}\in\mathcal{K}}\mathcal{A}_{\mathbf{k}}=\{h_{s,1},h_{s,2},\ldots,h_{s,\mu_s}\}$,$1\leq \mu_s\leq r$;
\State Divide $X_s$ into $\mu_s$ sub-packets, i.e., $X_s =\{X_{s,1}, X_{s,2},\ldots, X_{s,\mu_s}\}$;
\For {$l\in[\mu_s]$}
\State Server sends $X_{s,l}$ to relay $h_{s,l}$;
\State Relay $h_{s,l}$ broadcasts $X_{s,l}$ to its connecting users;
\EndFor
\EndFor
\EndProcedure
\end{algorithmic}
\end{algorithm}

From Algorithm \ref{alg:CPDA}, a $(K,F_1,Z,S)$ CPDA $\mathbf{P}$ can be explained intuitively as follows.
\begin{itemize}
\item Each row $j\in[F_1]$ represents the index of the $j^{\text{th}}$ packet of all files and each column $\mathbf{k}\in\mathcal{K}$ represents user $\mathbf{k}$. If $p_{j,\mathbf{k}}=*$, then user $\mathbf{k}$ has cached the $j^{\text{th}}$ packet of all files. So the condition C$1$ of Definition \ref{def-PDA} implies that each user caches $M=\frac{Z}{F_1}N$ files.
\item If $p_{j,\mathbf{k}}=s$ is an integer, it means that the $j^{\text{th}}$ packet of all files is not stored by user $\mathbf{k}$. Then the XOR of the requested packets indicated by $s$ is generated by the server at time slot $s$, denoted by $X_{s}$. If the set $\mathcal{I}_s$ defined in the condition C$4$ of Definition \ref{def-CPDA} is $\{h_{s,1},h_{s,2},\ldots,h_{s,\mu_s}\}$ where $1\leq \mu_s\leq r$, then $X_s$ is divided into $\mu_s$ sub-packets, i.e., $X_s =\{X_{s,1}, X_{s,2},\ldots, X_{s,\mu_s}\}$. Finally, for any $l\in[\mu_s]$, the server sends $X_{s,l}$ to relay $h_{s,l}$ and relay $h_{s,l}$ forwards $X_{s,l}$ to its connecting users. So the subpacketization is at most $rF_1$.
    The condition C$4$ of Definition \ref{def-CPDA} guarantees that user $\mathbf{k}$ can receive the whole message $X_{s}$, since user $\mathbf{k}$ is connected to each relay in $\mathcal{I}_s$. The condition C$3$ of Definition \ref{def-PDA} guarantees that user $\mathbf{k}$ can recover its requested packet indicated by $s$ from $X_{s}$, since it has cached all the other packets in the message $X_{s}$ except its requested one. The occurrence number of integer $s$ in $\mathbf{P}$, denoted by $g_s$, is the coded caching gain at time slot $s$, since the message $X_s$ is useful for $g_s$ users.
\item The condition C$2$ of Definition \ref{def-PDA} implies that the number of messages $X_s$ sent by the server is exactly $S$. If the size of $\mathcal{I}_s$ is a constant for each $s\in[S]$, assume that $|\mathcal{I}_s|=\mu$, and if the number of $\mathcal{I}_s$ containing $h$ is a constant for each relay $h\in[H]$, assume that $|\{\mathcal{I}_s|h\in \mathcal{I}_s, s\in[S]\}|=\nu$, then we have $\nu H=\mu S$. So the transmission load for each relay is $R=\nu\frac{1}{\mu F_1}=\frac{S}{HF_1}$.
\end{itemize}

\begin{lemma}(\cite{CLZW})
\label{th-1}
Given a $(K =u{H\choose r}, F_1,Z, S)$ CPDA for any positive integers $H$, $r$ and $u$ with $r<H$, we have an $(H, r,u,M, N)$ multiaccess combination network caching scheme with memory ratio $\frac{M}{N}= \frac{Z}{F_1}$, subpacketization $F\leq rF_1$. Moreover, if the size of $\mathcal{I}_s$, which is defined in Definition \ref{def-CPDA}, is a constant for each $s\in[S]$ and if the number of $\mathcal{I}_s$ containing $h$ is a constant for each relay $h\in[H]$, the transmission load for each relay is $R=\frac{S}{HF_1}$.
\hfill $\square$
\end{lemma}

In a PDA, a star not contained in any subarray shown as \eqref{form} in C$3$-b of Definition \ref{def-PDA}, is called useless. The authors in \cite{ZCWZ} pointed out that useless stars not only make no contribution to reducing the transmission load, but also result in a high subpacketization. If each column of a $(K,F,Z,S)$ PDA has $Z'$ useless stars, the authors in \cite{ZCWZ} improved the scheme in \cite{YTCC} by deleting all the useless stars and using an $[F,F-Z']_q$ maximum distance separable (MDS) code \cite{L} for some prime power $q$, and came up with a new coded caching scheme with smaller transmission load and subpacketization than the original scheme in \cite{YTCC} for the same number of users and memory ratio. In fact, this idea also works for the multiaccess combination network caching system.
\begin{lemma}\rm
\label{le-coded CPDA}
Given a $(K =u{H\choose r}, F_1,Z, S)$ CPDA $\mathbf{P}$ for any positive integers $H$, $r$ and $u$ with $r<H$, assume that there exist  $Z'$ useless stars in each column of $\mathbf{P}$. Then we have an $(H, r,u,M, N)$ multiaccess combination network caching scheme with memory ratio $\frac{M}{N}= \frac{Z-Z'}{F_1-Z'}$, subpacketization $F\leq r(F_1-Z')$. If the size of $\mathcal{I}_s$, which is defined in Definition \ref{def-CPDA}, is a constant for each integer $s\in[S]$ and if the number of $\mathcal{I}_s$ containing $h$ is a constant for each relay $h\in[H]$, the transmission load for each relay is $R=\frac{S}{H\left(F_1-Z'\right)}$.
\hfill $\square$
\end{lemma}
\begin{proof}
Assume that $\mathbf{P}$ is a $(K =u{H\choose r},F_1,Z,S)$ CPDA where each column has $Z'$ useless stars. Delete the $Z'$ useless stars in each column, we obtain a new array $\mathbf{P}'=(p'_{j,\mathbf{k}}|j\in[F_1], \mathbf{k}\in\mathcal{K})$. Clearly each column of $\mathbf{P}'$ has $Z'$ blanks, $Z-Z'$ stars and $F_1-Z$ integers.
Based on $\mathbf{P}'$, we modify the placement strategy in Algorithm \ref{alg:CPDA} as follows: the server divides each file into $F_1-Z'$ equal-sized packets and then encodes them using an $[F_1,F_1-Z']_q$ MDS code for some prime power $q$. The resulting encoded packets are denoted by $W_{n,1}$, $W_{n,2}$,$\ldots$, $W_{n,F_1}$ for each file $W_{n}$ where $n\in[N]$. Using the caching strategy in Lines 3-5 in Algorithm \ref{alg:CPDA}, each user $\mathbf{k}$ caches $Z_{\mathbf{k}}=\{W_{n,j}|p'_{j,\mathbf{k}}=*, j\in[F_1], n\in[N]\}$. Clearly, the memory ratio of each user is $\frac{M}{N}=\frac{Z-Z'}{F_1-Z'}$. Now let us consider its subpacketization and transmission load for each relay. For any request vector $\mathbf{d}$ in the delivery phase, we use the same delivery strategy as in Algorithm \ref{alg:CPDA}. Then each user can get exactly $F_1-Z$ required coded packets. From the property of an $[F_1,F_1-Z']_q$ MDS code, each user can recover its requested file. Since the size of $\mathcal{I}_s$ is a constant for each integer $s\in[S]$ and the number of $\mathcal{I}_s$ containing $h$ is a constant for each relay $h\in[H]$, the transmission load for each relay is $R=\frac{S}{H\left(F_1-Z'\right)}$.
\end{proof}

\begin{remark}\rm
 Given a $(K,F_1,Z,S)$ CPDA $\mathbf{P}$, if each column of $\mathbf{P}$ has $Z'$ useless stars, then the scheme in Lemma \ref{le-coded CPDA} has smaller memory ratio and subpacketization than the scheme in Lemma \ref{th-1}, since $\frac{Z}{F_1}>\frac{Z-Z'}{F_1-Z'}$ and $F_1>F_1-Z'$ always hold for any positive integer $Z'<Z$. It is worth noting that for the scheme in Lemma \ref{le-coded CPDA}, the operation field must be $O(rF_1)$, hence the size of each packet must be proximately of length $\log_2(rF_1)$ bits. This implies that the size of each file in the server must be more than $\log_2(rF_1)(r(F_1-Z'))$.
 \hfill $\square$
\end{remark}
\section{The schemes for the multiaccess combination network}
\label{sec-first}
In this section, we will propose three coded caching schemes for the multiaccess combination network caching system. The first one is obtained directly from the idea in \cite{ZY2}. The second one is from a direct construction of CPDA via combinatorial design theory and the third one is from a hybrid construction of CPDA based on the CPDA directly constructed before and the MN PDA.
\subsection{The first scheme from the idea in \cite{ZY2}}
\label{subsect-MDS}
Since the scheme in \cite{ZY2} can be used to deal with any relay networks with arbitrary topologies, it also works for the multiaccess combination network caching setting. Specifically, from the idea in \cite{ZY2} we can get an $(H,r,u,M,N)$ multiaccess combination network caching scheme as follows. Firstly, the server divides each file into $r$ packets and then encodes them into $H$ coded packets by an $[H,r]_q$ MDS code for some prime power $q$. Denote the $h^{\text{th}}$ coded packet of all files by $\mathcal{W}^{(h)}=\{W^{(h)}_n\ |\ n\in[N]\}$. It is worth noting that the size of each coded packet is $\frac{B}{r}$ bits. Secondly, for each relay $h$, by taking the $h^{\text{th}}$ coded packet of all files  $\mathcal{W}^{(h)}$ as the file library and using the $(\widetilde{K}=u{H-1\choose r-1},M,N)$ MN scheme, the server sends all the required coded signals to relay $h$, then relay $h$ forwards them to its connecting users. Let $t=\frac{\widetilde{K}M}{N}$, then each user caches exactly $\frac{r{\widetilde{K}-1\choose t-1}}{r{\widetilde{K}\choose t}}NB=MB$ bits, since each user is just connected to $r$ relays.
Thus the transmission load for each relay is $R=\frac{\widetilde{K}(1-M/N)}{r(\widetilde{K}M/N+1)}$. That is the following result.
\begin{lemma}(ZY scheme)
\label{le-MN-PDA}
For any positive integers $H$, $r$, $u$, $M$ and $N$ with $r<H$ and $M<N$, let $\widetilde{K}=u{H-1\choose r-1}$, if $t=\frac{\widetilde{K}M}{N}$ is an integer, we have an $(H,r,u,M,N)$ multiaccess combination network caching scheme with subpacketization
$F=r{\widetilde{K}\choose t}$ and transmission load for each relay $R=\frac{\widetilde{K}(1-\frac{M}{N})}{r(t+1)}$.
\hfill $\square$
\end{lemma}

It is easy to verify that the transmission load for each relay of ZY scheme is approximately $\frac{1}{r}$ of the transmission load of MN scheme. However, the subpacketization of ZY scheme is exponential order with $\widetilde{K}=u{H-1\choose r-1}$, which may lead to high implementation complexity and infeasibility in reality. Therefore, reducing the subpacketization while slightly increasing the transmission load for each relay as the tradeoff is very meaningful.
\subsection{The second scheme via direct construction of CPDA}
\label{subsect-direct}
In this section, we propose a direct construction of CPDA by means of Combinatorial Design Theory, which leads to a multiaccess combination network caching scheme with much lower subpacketization than ZY scheme in Lemma \ref{le-MN-PDA}.

\begin{construction}
\label{constr-1}
For any positive integers $H$, $r$, $a$, $\omega$ and $\lambda$ with $\max\{\omega,\lambda\}\leq r<H$
and $a<H$, let the row index set $\mathcal{F}$ and the column index set $\mathcal{K}$ be
\begin{eqnarray}
\begin{split}
\label{eq-row-column-labels}
\mathcal{F}&=\left\{{\bf f}=(f_1, f_2,\ldots, f_{H})\in \{0,1\}^H\ |\  wt({\bf f})=a\right\},&\\
\mathcal{K}&=\left\{\mathbf{k}=(\mathcal{T},\mathbf{b}) \ \left|\  \mathcal{T}\in {[H]\choose r}, {\bf b}\in \{0,1\}^r, \text{wt}({\bf b})=r-\omega\right\}\right.,&
\end{split}
\end{eqnarray}
respectively.
An ${H\choose a}\times {r\choose \omega}{H\choose r}$ array $\mathbf{P}=(p_{\mathbf{f},\mathbf{k}})_{ \mathbf{f}\in \mathcal{F},\mathbf{k}\in \mathcal{K}}$ with $\mathcal{T}=\{ \delta_1, \delta_2, \ldots, \delta_{r}\}\subset[H]$, $\delta_1 <\delta_2<\cdots <\delta_{r}$ and $\mathbf{b}=(b_1,b_2,\ldots,b_{r})$, is defined as
 \begin{eqnarray}
 \label{eq-defin-array}
p_{\mathbf{f},\mathbf{k}}=\left\{\begin{array}{cccccc}
(\mathbf{e},\mathcal{C}) &\ \ \ \hbox{if}\ \ d(\mathbf{f}|_{\mathcal{T}},\mathbf{b})=r-\lambda, \\[0.2cm]
* &\ \ \ \hbox{Otherwise},\ \ \\
\end{array}\right.
\end{eqnarray}
where $\mathbf{e}=(e_1,e_2,\dots,e_{H})\in\{0,1\}^H$ is defined as
\begin{eqnarray}
\label{eq-entry}
e_i=\left\{\begin{array}{cccccc}
b_v &\ \ \ \hbox{if}\ \ i=\delta_v, v\in[r],\\[0.2cm]
f_i &\ \ \ \hbox{Otherwise}\ \ \\
\end{array}\right.
\end{eqnarray}
and
\begin{equation}
\label{eq-c}
\mathcal{C}=\{\delta_v\ |\ f_{\delta_v}=b_v, v\in[r]\}\subseteq\{\delta_1,\delta_2,\ldots,\delta_{r}\}.
\end{equation}
Clearly, we have $|\mathcal{C}|=\lambda$ from \eqref{eq-defin-array}.
\hfill $\square$
\end{construction}

\begin{example}
\label{exam-1}
When $H = 4$, $a=2$, $r = 2$, $\omega=1$ and $\lambda=1$, the array generated by Construction \ref{constr-1} is shown in Table \ref{tab-1}.
\begin{table}[H]

  \centering
\renewcommand\arraystretch{1}
\setlength{\tabcolsep}{0.8mm}{
  \fontsize{8}{10}\selectfont
  \caption{The $(12,6,4,8)$ CPDA with $H=4$, $a=r=2$, $\omega=1$ and $\lambda=1$. \label{tab-1}}
    \begin{tabular}{|c|c|c|c|c|c|c|c|c|c|c|c|c|c|c|c|c|c|c|c|c|}
    \hline
    \multirow{2}*{\diagbox{${\bf f}$}{$(\mathcal{T},\mathbf{b})$}}& \multicolumn{2}{c|}{\{1,2\}}&\multicolumn{2}{c|}{\{1,3\}}&\multicolumn{2}{c|}{ \{1,4\}}&\multicolumn{2}{c|}{\{2,3\}}&\multicolumn{2}{c|}{\{2,4\}}&\multicolumn{2}{c|}{\{3,4\}}\cr\cline{2-13}

    &(0,1)&(1,0)&(0,1)&(1,0)&(0,1)&(1,0)&(0,1)&(1,0)&(0,1)&(1,0)&(0,1)&(1,0)\cr
    \hline

    (0,0,1,1)&0111,1&	1011,2&	*&	*	&*&	*&	*&	*&	*	&*&	0001,4&	0010,3\cr\hline
    (0,1,0,1)&*	&*&	0111,1	&1101,3	&*	&*	&*	&*	&0001,4	&0100,2&	*	&*\cr\hline
    (0,1,1,0)&*	&*	&*&	*	&0111,1	&1110,4	&0010,3	&0100,2	&*	&*	&*	&*\cr\hline
    (1,0,0,1)&*	&*	&*	&*	&0001,4	&1000,1&1011,2&	1101,3&	*&	*&	*&	*\cr\hline
    (1,0,1,0)&*	&*	&0010,3	&1000,1	&*&	*	&*	&*	&1011,2	&1110,4	&*	&*\cr\hline
    (1,1,0,0)&0100,2&	1000,1&	*&	*&	*&	*	&*&	*&	*&	*	&1101,3&	1110,4\cr\hline
    \end{tabular} }
\end{table}For instance, when ${\bf f}=(0, 0,1, 1)$ and $(\mathcal{T},{\bf b})=(\{1,2\},(1,0))$, since $\text{d}({\bf f}|_{\{1,2\}},(1,0))=\text{d}((0,0),(1,0))=1=r-\lambda$, we have $\mathcal{C}=\{2\}$ from \eqref{eq-c}, then from \eqref{eq-defin-array} and \eqref{eq-entry} we have $p_{{\bf f},(\mathcal{T},{\bf b})}=((1,0,1,1),\{2\})$, abbreviated as $1011,2$. It is easy to verify that the array in Table \ref{tab-1} is a $(12, 6, 4, 8)$ CPDA, which leads to an $(H,r,u,M,N)$ multiaccess combination network caching scheme with $u={r\choose \omega}=2$,  memory ratio $\frac{M}{N}=\frac{2}{3}$, subpacketization $F=6$ and transmission load for each relay $R=\frac{8}{6}\cdot\frac{1}{4}=\frac{1}{3}$. In this case, the subpacketization and transmission load for each relay of ZY scheme in Lemma \ref{le-MN-PDA} are $F_{ZY}=30$ and $R_{ZY}=\frac{1}{5}$ respectively.
\hfill $\square$
\end{example}

In general, the array generated by Construction \ref{constr-1} is a CPDA, where each column has the same number of useless stars. Hence, from Lemma \ref{le-coded CPDA} we have the following result, whose proof could be found in Appendix \ref{appendix-th-2}.
\begin{theorem}(Scheme A)
\label{th-2}
For any positive integers $H$, $r$, $a$, $\omega$ and $\lambda$ with $\max\{\omega,\lambda\}\leq r<H$
and $a<H$, there exists a $({r\choose \omega}{H\choose r},{H\choose a},Z,S)$ CPDA with
$Z={H\choose a}-\sum\limits_{\lambda_1\in [x:y]}{\omega\choose \lambda_1}{r-\omega\choose \lambda-\lambda_1}{H-r\choose a-\omega+2\lambda_1-\lambda}$ and  $S=\sum\limits_{\lambda_1\in [x:y]}{H\choose a+r-2\omega+2\lambda_1-\lambda}{H-(a+r-2\omega+2\lambda_1-\lambda)\choose\lambda_1}$ ${a+r-2\omega+2\lambda_1-\lambda\choose \lambda-\lambda_1}$,
where $x=\max\{0,\lambda-r+\omega\}$, $y=\min\{\omega,\lambda\}$. It can lead to an $(H,r,u={r\choose \omega},M,N)$ multiaccess combination network caching scheme with memory ratio, subpacketization and transmission load for each relay being
\begin{eqnarray}
\frac{M}{N}&=&1-\frac{\sum\limits_{\lambda_1\in [x:y]}{\omega\choose \lambda_1}{r-\omega\choose \lambda-\lambda_1}{H-r\choose a-\omega+2\lambda_1-\lambda}}{{H\choose a}-\sum\limits_{\lambda'\in[0:\lambda-1]}\sum\limits_{\lambda_1\in[x':y']}{\omega\choose \lambda_1}\times {r-\omega\choose \lambda'-\lambda_1}\times{H-r\choose a-\omega+2\lambda_1-\lambda'}},\label{eq-memory-ratio} \\
F&=&\lambda\left({H\choose a}-\sum\limits_{\lambda'\in[0:\lambda-1]}\sum\limits_{\lambda_1\in[x':y']}{\omega\choose \lambda_1}\times {r-\omega\choose \lambda'-\lambda_1}\times{H-r\choose a-\omega+2\lambda_1-\lambda'}\right),\label{eq-subpacketization} \\
R&=&\frac{\sum\limits_{\lambda_1\in[x:y]}{H\choose a+r-2\omega+2\lambda_1-\lambda}
{H-(a+r-2\omega+2\lambda_1-\lambda)\choose \lambda_1}
{a+r-2\omega+2\lambda_1-\lambda\choose \lambda- \lambda_1}}{H\left({H\choose a}-\sum\limits_{\lambda'\in[0:\lambda-1]}\sum\limits_{\lambda_1\in[x':y']}{\omega\choose \lambda_1}\times {r-\omega\choose \lambda'-\lambda_1}\times{H-r\choose a-\omega+2\lambda_1-\lambda'}\right)}\label{eq-rate}
\end{eqnarray}
respectively, where $x'=\max\{0,\lambda'-r+\omega\}$ and $y'=\min\{\omega,\lambda'\}$.
\hfill $\square$
\end{theorem}

\begin{remark}
\label{rem-1}
\begin{itemize}
\item When $\omega=0$, from Theorem \ref{th-2} we can obtain an $({H\choose r},{H\choose a}, {H\choose a}-{r\choose \lambda}{H-r\choose a-\lambda}$, ${H\choose a+r-\lambda}{a+r-\lambda \choose \lambda})$ CPDA which is exactly the strongly coloring PDA in \cite{YTCC} by the fact ${H\choose a+r-\lambda}{a+r-\lambda \choose \lambda}={H\choose a+r-2\lambda}{H-(a+r-2\lambda) \choose \lambda}$.
\item The row index set $\mathcal{F}$ in \eqref{eq-row-column-labels} is determined by the value of $a$. In fact, $\mathcal{F}$ can be chosen more flexibly. For example, let $\mathcal{F}=\{(f_{1},f_{2},\ldots,f_{H})\ |\ f_{1},f_{2},\ldots,f_{H-1}\in[0:q-1], f_{H}=mod(\sum_{i=1}^{H-1} f_i,q)\}$ and $\mathcal{K}={[H]\choose r}\times [0:q-1]^r$ for any positive integer $q\geq 2$, when $H=4$, $r=2$ and $\lambda=1$, according to \eqref{eq-defin-array}, \eqref{eq-entry} and \eqref{eq-c}, a $(24,8,4,32)$ CPDA can be obtained, which is shown in Table \ref{tab-2}.
    \begin{table}[H]
\caption{The $(24,8,4,32)$ CPDA. \label{tab-2}}
  \centering
\renewcommand\arraystretch{1}
\setlength{\tabcolsep}{0.4mm}{
  \fontsize{6}{8}\selectfont
   \begin{tabular}{|c|c|c|c|c|c|c|c|c|c|c|c|c|c|c|c|c|c|c|c|c|c|c|c|c|c|c|}
    \hline
    \multirow{2}*{\diagbox{${\bf f}$}{$(\mathcal{T},\mathbf{b})$}}&
    \multicolumn{4}{c|}{\{1,2\}}&\multicolumn{4}{c|}{\{1,3\}}&\multicolumn{4}{c|}{ \{1,4\}}&\multicolumn{4}{c|}{\{2,3\}}&\multicolumn{4}{c|}{\{2,4\}}&\multicolumn{4}{c|}{ \{3,4\}}\cr\cline{2-25}
    &(0,0)&(0,1)&(1,0)&(1,1)&(0,0)&(0,1)&(1,0)&(1,1)&(0,0)&(0,1)&(1,0)&(1,1)&(0,0)&(0,1)&(1,0)&(1,1)&(0,0)&(0,1)&(1,0)&(1,1)&(0,0)&(0,1)&(1,0)&(1,1)\cr\hline
(0,0,0,0)&*&0100,1	&1000,2&*&	*&0010,1&1000,3&*&*&0001,1&	1000,4&	*&*	&0010,2	&0100,3	&*&	*&	0001,2&	0100,4&	*&	*&	0001,3&	0010,4&	*\\ \hline
(0,0,1,1)&	*&	0111,1	&1011,2	&*	&0001,1	&*	&*	&1011,3	&0010,1&	*	&*	&1011,4&	0001,2	&*	&*	&0111,3	&0010,2&	*	&*	&0111,4	&*&	0001,4	&0010,3	&* \\ \hline
(0,1,0,1)&	0001,1&	*	&*&	1101,2&	*&	0111,1&	1101,3&	*	&0100,1	&*	&*&	1101,4	&0001,3	&*	&*&	0111,2	&*	&0001,4&	0100,2&	*	&0100,3&	*&	*	&0111,4 \\ \hline
(0,1,1,0)&	0010,1&	*	&*	&1110,2&	0100,1&	*	&*	&1110,3	&*	&0111,1	&1110,4	&*	&*	&0010,3	&0100,2	&*	&0010,4	&*	&*	&0111,2	&0100,4	&*	&*	&0111,3 \\ \hline
(1,0,0,1)&	0001,2&	*	&*	&1101,1	&0001,3&	*&	*&	1011,1&	*	&0001,4	&1000,1&	*&	*&	1011,2	&1101,3	&*	&1000,2&	*	&*	&1101,4&	1000,3	&*	&*	&1011,4\\ \hline
(1,0,1,0)&	0010,2	&*	&*	&1110,1	&*	&0010,3&	1000,1	&*	&0010,4	&*&	*	&1011,1&	1000,2&	*	&*	&1110,3	&*&	1011,2	&1110,4	&*&	1000,4&	*	&*	&1011,3 \\ \hline
(1,1,0,0)&	*&	0100,2&	1000,1	&*	&0100,3	&*&	*	&1110,1&	0100,4&	*&	*	&1101,1&	1000,3	&*	&*	&1110,2&	1000,4&	*	&*	&1101,2&	*&	1101,3&	1110,4	&* \\ \hline
(1,1,1,1)&	*&	0111,2&	1011,1	&*	&*&	0111,3	&1101,1&	*	&*	&0111,4&	1110,1	&*	&*	&1011,3	&1101,2	&*	&*	&1011,4&	1110,2	&*	&*	&1101,4	&1110,3&	*\\ \hline
    \end{tabular}
}
\end{table}
\item From the {\em footnote} \ref{foot:1} in  Appendix \ref{appendix-th-2}, the scheme from the CPDA generated by Construction \ref{constr-1} omits the coded caching gain among the users who are connected to the same $r$ relays.
\end{itemize}
\hfill $\square$
\end{remark}

\subsection{The third scheme via hybrid construction of CPDA}
\label{section-concatenating}
From Theorem \ref{th-2}, an $(H,r,u={r\choose \omega},M,N)$ multiaccess combination network caching scheme can be obtained. However, it has a limitation of the parameter $u$, i.e., $u={r\choose \omega}$. In this subsection, we will propose a hybrid construction of CPDA for any parameter $u$ based on a PDA and a CPDA. Let us use an example to illustrate the main idea of the construction.
\begin{example}
\label{exam-3}
Given the $(3,3,1,3)$ CPDA $\mathbf{P}$ with $H=3$, $r=2$, $u=1$, and the $(2,2,1,1)$ PDA $\mathbf{A}$ in \eqref{PA},
\begin{figure}
\begin{eqnarray}
\label{PA}
\mathbf{P}=
{%
\begin{blockarray}{cccc}
\{1,2\}&\{1,3\}&\{2,3\} \\[-0.2cm]
\begin{block}{(ccc)c}
*&	1&	2 &1\\[-0.2cm]
1&	*&	3 &2\\[-0.2cm]
2&	3&	*&3\\[-0.2cm]
\end{block}
\end{blockarray}}\ \ \ \ \ \ \ \ \ \ \ \mathbf{A}=
{%
\begin{blockarray}{ccc}
 1& 2\\[-0.2cm]
\begin{block}{(cc)c}
*&	1&1 \\[-0.2cm]
1&	*&2\\[-0.2cm]
\end{block}
\end{blockarray}}
\end{eqnarray}
\end{figure}
by replacing each entry in $\mathbf{P}$, i.e., $p_{j_1,\mathbf{k}_1}$ where $j_1\in[3]$, $\mathbf{k}_1\in {[3]\choose 2}$, with an array $\mathbf{A}+p_{j_1,\mathbf{k}_1}=(a_{j_2,k_2}+p_{j_1,\mathbf{k}_1})_{j_2\in[2],k_2\in [2]}$, where $x+*=*$ for any integer $x$, we can obtain the following $(6,6,3,3)$ CPDA $\mathbf{L}$ with $H=3,r=2$ and $u=2$.
\begin{eqnarray*}
\mathbf{L}={%
\begin{blockarray}{ccccccc}
(\{1,2\},1)&(\{1,2\},2)&(\{1,3\},1)&(\{1,3\},2)&(\{2,3\},1)&(\{2,3\},2) \\
\begin{block}{(cccccc)c}
* & * & * & 1 & * & 2&(1,1)\\[-0.2cm]
*&	*&	1&	*&	2&	*&(1,2)\\[-0.2cm]
*&	1&	*&	*&	*&	3 &(2,1)\\[-0.2cm]
1&	*&	*&	*&	3&	* &(2,2)\\[-0.2cm]
*&	2&	*&	3&	*&	* &(3,1)\\[-0.2cm]
2&	*&	3&	*&	*&	*&(3,2) \\[-0.2cm]
\end{block}
\end{blockarray}}
\end{eqnarray*} It is easy to verify that the scheme from $\mathbf{L}$ creates broadcast opportunities among the users who are connected to the same $r$ relays, since $l_{(2,2),(\{1,2\},1)}=l_{(2,1),(\{1,2\},2)}=1$.
\hfill $\square$
\end{example}

Mathematically, the construction is given as follows.
\begin{construction}
\label{constr-2}
Suppose that there exists a $(K_1=u{H\choose r},F_1,Z_1,S_1)$ CPDA $\textbf{P}=(p_{j_1,{\bf k}_1})$, $j_1\in[F_1]$, ${\bf k}_1\in {[H]\choose r}\times[u]$, and a $(K_2,F_2,Z_2,S_2)$ PDA $\textbf{A}=(a_{j_2,k_2})$, $j_2\in[F_2]$, $k_2\in [K_2]$, the array $\textbf{L}=(l_{(j_1,j_2),({\bf k}_1,k_2)})$, $(j_1,j_2)\in[F_1]\times [F_2]$, $({\bf k}_1,k_2)\in({[H]\choose r}\times[u])\times [K_2]$, is defined as
\begin{eqnarray}
\label{const-2}
l_{(j_1,j_2),({\bf k}_1,k_2)}=\left\{\begin{array}{ccc}
a_{j_2,k_2}+(p_{j_1,{\bf k}_1}-1)S_2&\ \ \ \hbox{if}\ \ p_{j_1,{\bf k}_1}\neq * , a_{j_2,k_2}\neq*,\\
* &\ \ \hbox{otherwise}.
\end{array}\right.
\end{eqnarray}
\hfill $\square$
\end{construction}
Let us return to Example \ref{exam-3}, when $j_1=1,{\bf k}_1=\{1,2\}, j_2=1,k_2=1$, since $p_{j_1,{\bf k}_1}=p_{1,\{1,2\}}=*$, we have $l_{(j_1,j_2),({\bf k}_1,k_2)}=*$ from \eqref{const-2}; when $j_1=2,{\bf k}_1=\{1,2\}, j_2=2,k_2=1$, since $p_{j_1,{\bf k}_1}=p_{2,\{1,2\}}=1$ and $a_{j_2,k_2}=a_{2,1}=1$, we have $l_{(j_1,j_2),({\bf k}_1,k_2)}=a_{j_2,k_2}+(p_{j_1,{\bf k}_1}-1)S_2=1$ from \eqref{const-2}.

By Construction \ref{constr-2}, we have the following result, whose proof could be found in Appendix \ref{appendix-th-3}.
\begin{theorem}
\label{th-3}
Given a $(K_1=u{H\choose r},F_1,Z_1,S_1)$ CPDA $\mathbf{P}=\left(p_{j_1,{\bf k}_1}\right)_{j_1\in[F_1],{\bf k}_1\in{[H]\choose r}\times [u]}$ and a $(K_2,F_2,Z_2,S_2)$ PDA $\mathbf{A}$, there always exists a $(uK_2{H\choose r},$ $F_1F_2,Z_1F_2+(F_1-Z_1)Z_2,S_1S_2)$ CPDA $\mathbf{L}$.
If each set $\mathcal{I}_{s_1}$ for $\mathbf{P}$ (denoted by $\mathcal{I}_{s_1}^{(\mathbf{P})}$) where $s_1\in[S_1]$ has the same size $\mu$ and the number of $\mathcal{I}_{s_1}^{(\mathbf{P})}$ containing $h$ is a constant for each $h\in[H]$,
and if there are $Z'_1$ useless stars in each column of $\mathbf{P}$, then $\mathbf{L}$ leads to an $(H,r,uK_2,M,N)$ multiaccess combination network caching scheme with memory ratio $\frac{M}{N}=\frac{Z_1-Z'_1}{F_1-Z'_1}+\left(1-\frac{Z_1-Z'_1}{F_1-Z'_1}\right)\frac{Z_2}{F_2}$, subpacketization $F=\mu(F_1-Z'_1)F_2$ and transmission load for each relay $R=\frac{S_1S_2}{H(F_1-Z'_1)F_2}$.
\hfill $\square$
\end{theorem}

From Theorem \ref{th-3}, the following result can be obtained based on the CPDA in Theorem \ref{th-2} and the $\left(K_2,{K_2\choose t_2},{K_2-1\choose t_2-1},{K_2\choose t_2+1}\right)$ MN PDA for any $K_2,t_2\in\mathbb{Z}^+$ with $t_2<K_2$.
\begin{theorem}(Scheme B)
\label{th-4}
For any positive integers $H$, $r$, $a$, $\omega$, $\lambda$, $t_2$ and $K_2$ with $\max\{\omega,\lambda\}\leq r<H$, $a<H$ and $t_2<K_2$, there exists an $(H,r,u=K_2{r\choose \omega},M,N)$ multiaccess combination network caching scheme with the memory ratio, subpacketization and transmission load for each relay being
\begin{equation}
\label{eq-memory-ratio-concatenating}
\frac{M}{N}=1-\frac{K_2-t_2}{K_2}\cdot\frac{\sum\limits_{\lambda_1\in [x:y]}{\omega\choose \lambda_1}{r-\omega\choose \lambda-\lambda_1}{H-r\choose a-\omega+2\lambda_1-\lambda}}{{H\choose a}-\sum\limits_{\lambda'\in[0:\lambda-1]}\sum\limits_{\lambda_1\in[x':y']}{\omega\choose \lambda_1}\times {r-\omega\choose \lambda'-\lambda_1}\times{H-r\choose a-\omega+2\lambda_1-\lambda'}},
\end{equation}
\begin{equation}
\label{eq-subpacketization-concatenating}
F=\lambda{K_2\choose t_2}\left({H\choose a}-\sum\limits_{\lambda'\in[0:\lambda-1]}\sum\limits_{\lambda_1\in[x':y']}{\omega\choose \lambda_1}\times {r-\omega\choose \lambda'-\lambda_1}\times{H-r\choose a-\omega+2\lambda_1-\lambda'}\right),
\end{equation}
\begin{equation}
\label{eq-rate-concatenating}
R=\frac{K_2-t_2}{t_2+1}\cdot\frac{\sum\limits_{\lambda_1\in[x:y]}{H\choose a+r-2\omega+2\lambda_1-\lambda}
{H-(a+r-2\omega+2\lambda_1-\lambda)\choose \lambda_1}
{a+r-2\omega+2\lambda_1-\lambda\choose \lambda- \lambda_1}}{H\left({H\choose a}-\sum\limits_{\lambda'\in[0:\lambda-1]}\sum\limits_{\lambda_1\in[x':y']}{\omega\choose \lambda_1}\times {r-\omega\choose \lambda'-\lambda_1}\times{H-r\choose a-\omega+2\lambda_1-\lambda'}\right)},
\end{equation}
respectively, where $x=\max\{0,\lambda-r+\omega\}$, $y=\min\{\omega,\lambda\}$, $x'=\max\{0,\lambda'-r+\omega\}$ and $y'=\min\{\omega,\lambda'\}$.
\hfill $\square$
\end{theorem}

\section{Performance analysis}
\label{sect-performance}
In this section, we compare Scheme A in Theorem \ref{th-2} and Scheme B in Theorem \ref{th-4} with ZY scheme in Lemma \ref{le-MN-PDA} and the naive scheme of repeatedly using the scheme in \cite{ZCWZ} (referred to as Scheme C).

\subsection{Analytic Comparison of Scheme A,B with ZY scheme}
\label{subsect-Th2-MDS}
Since it is difficult to give an analytic comparison of Scheme A and Scheme B with ZY scheme for general parameters $\lambda$ and $\omega$, we focus on the case of $\lambda=1, \omega=1$. In this case, the memory ratio, subpacketization and transmission load for each relay of Scheme A are
\begin{eqnarray*}
\frac{M}{N}&=&1-\frac{{H-r\choose a}+\left(r-1\right){H-r\choose a-2}}{{H\choose a}-{H-r\choose a-1}},\ \ \ \ \ \ F_{Th1}={H\choose a}-{H-r\choose a-1},\\
R_{Th1}&=&\frac{(1-\frac{M}{N})}{H}\cdot\frac{{H\choose a+r-1}\left(H-a-r+1\right)+{H\choose a+r-3}\left(a+r-3\right)}{{H-r\choose a}+\left(r-1\right){H-r\choose a-2}}
\end{eqnarray*}
respectively from \eqref{eq-memory-ratio}, \eqref{eq-subpacketization} and \eqref{eq-rate}; from \eqref{eq-memory-ratio-concatenating}, \eqref{eq-subpacketization-concatenating} and \eqref{eq-rate-concatenating}, Scheme B has memory ratio, subpacketization and transmission load for each relay as follows,

\begin{eqnarray*}
\label{eq-theorem-4}
\begin{split}
\frac{M}{N}&=1-\frac{K_2-t_2}{K_2}\cdot\frac{{H-r\choose a}+\left(r-1\right){H-r\choose a-2}}{{H\choose a}-{H-r\choose a-1}},\ \ \ \ \ \ F_{Th3}={K_2\choose t_2}\left({H\choose a}-{H-r\choose a-1}\right),&\\
R_{Th3}&=\frac{K_2(1-\frac{M}{N})}{H(t_2+1)}\cdot\frac{{H\choose a+r-1}\left(H-a-r+1\right)+{H\choose a+r-3}\left(a+r-3\right)}{{H-r\choose a}+\left(r-1\right){H-r\choose a-2}}.&
\end{split}
\end{eqnarray*}
While ZY scheme has subpacketization and transmission load for each relay as follows,
\begin{eqnarray*}
F_{ZY}=r{u{H-1\choose r-1}\choose \frac{u{H-1\choose r-1}M}{N}},\ \ \ \ R_{ZY}=\frac{u{H-1\choose r-1}(1-\frac{M}{N})}{r\left(\frac{u{H-1\choose r-1}M}{N}+1\right)}.
\end{eqnarray*}
So we have
\begin{eqnarray*}
\label{Th2-R-2}
\begin{split}
\frac{R_{Th1}}{R_{ZY}}&=\left(\frac{rM}{HN}+\frac{1}{K}\right)\frac{{H\choose a+r-1}\left(H-(a+r-1)\right)+{H\choose a+r-3}\left(a+r-3\right)}{{H-r\choose a}+\left(r-1\right){H-r\choose a-2}}&\\
&=\left(\frac{rM}{HN}+\frac{1}{K}\right)H\left(H-1\right)\cdots\left(H-r+1\right)\frac{\frac{\left(H-a-r+3\right)\left(H-a-r+2\right)\left(H-a-r+1\right)}{\left(a+r-1\right)!}+\frac{1}{\left(a+r-4\right)!}}{\frac{\left(H-a-r+3\right)\left(H-a-r+2\right)\left(H-a-r+1\right)}{a!}+\left(r-1\right)\frac{H-r-a+3}{\left(a-2\right)!}}&\\
&<\left(\frac{rM}{HN}+\frac{1}{K}\right)\frac{H\left(H-1\right)\cdots\left(H-r+1\right)}{\left(a+r-1\right)\cdots\left(a+1\right)}&
\end{split}
\end{eqnarray*}
and
\begin{equation*}
\label{Th4-R-2}
\frac{R_{Th3}}{R_{ZY}}<\frac{K_2}{t_2+1}\left(\frac{rM}{HN}+\frac{1}{K}\right)\frac{H\left(H-1\right)\cdots\left(H-r+1\right)}{\left(a+r-1\right)\cdots\left(a+1\right)}
\end{equation*}
when $H$ is large enough.
And we have
\begin{eqnarray*}
\frac{F_{Th1}}{F_{ZY}}&=\frac{{H\choose a}-{H-r\choose a-1}}{r{u{H-1\choose r-1}\choose \frac{u{H-1\choose r-1}M}{N}}}<\frac{{H\choose a}-{H-r\choose a-1}}{\left(\frac{N}{M}\right)^{\frac{r{H-1\choose r-1}M}{N}}}\ \ \ \text{and}\ \ \ \frac{F_{Th3}}{F_{ZY}}&<\frac{{K_2\choose t_2}\left({H\choose a}-{H-r\choose a}\right)}{\left(\frac{N}{M}\right)^{\frac{K_2r{H-1\choose r-1}M}{N}}}.
\end{eqnarray*}
Clearly, for Scheme A and Scheme B, the transmission load for each relay is just multiplied while the subpacketization is reduced exponentially compared to ZY scheme. Moreover, it is easy to verify that $\left(\frac{rM}{HN}+\frac{1}{K}\right)\frac{H\left(H-1\right)\cdots\left(H-r+1\right)}{\left(a+r-1\right)\cdots\left(a+1\right)}\ll \frac{r{H-1\choose r-1}M}{N}$ and $\frac{K_2}{t_2+1}\left(\frac{rM}{HN}+\frac{1}{K}\right)\frac{H\left(H-1\right)\cdots\left(H-r+1\right)}{\left(a+r-1\right)\cdots\left(a+1\right)}\ll \frac{K_2r{H-1\choose r-1}M}{N}$, i.e., in Scheme A and Scheme B the growth multiple of transmission loads for each relay are much less than the reduction exponent of the subpacketization.
\subsection{Numerical comparison}
The numerical comparison of Scheme A and Scheme B with ZY scheme is given in Table \ref{tab:2}. It is easy to see from Table \ref{tab:2} that $\frac{R_{Th1}}{R_{ZY}}\ll \ln\frac{F_{ZY}}{F_{Th1}}$ and $\frac{R_{Th3}}{R_{ZY}}\ll \ln\frac{F_{ZY}}{F_{Th3}}$, which coincides with the analytic comparison.
\begin{table}[H]
\caption{Numerical comparison of Scheme A and Scheme B when $\omega=1,\lambda=1$ with ZY scheme}
	\centering
	\begin{tabular}{|l|c|c|c|c|c|c|}\hline
			$(H,r,a)$&K&$\frac{M}{N}$&$\frac{R_{Th1}}{R_{ZY}}$&$\ln\frac{F_{ZY}}{F_{Th1}}$\\\hline
			15,2,4&210&0.27&1.59&11.75\\
			16,3,4&1680&0.43&8.19&297.35\\
			18,3,6&2448&0.50&5.98&391.00\\
			20,3,12&3420&0.56&4.03&487.29\\\hline
		\end{tabular}	\begin{tabular}{|l|c|c|c|c|c|c|}\hline
			$(H,r,a,K_2,t_2)$&K&$\frac{M}{N}$&$\frac{R_{Th3}}{R_{ZY}}$&$\ln\frac{F_{ZY}}{F_{Th3}}$\\\hline
		     20,3,3,8,1& 27360& 0.38& 53.68& 474.64\\
            20,3,3,6,1& 20520& 0.41& 43.39& 484.46\\
            20,3,3,5,1& 17100 &0.43 &38.24& 490.37\\
            20,3,3,3,1& 10260& 0.53&  27.95&497.17\\\hline
	\end{tabular}
		\label{tab:2}
\end{table}Finally, when $H=14$, $r=4$, $u=6$ and $K=u{H\choose r}=6006$, the memory-load and memory-subpacketization tradeoffs are given in Fig.\ref{fig:numerical}.
From Fig.\ref{fig:numerical}(a), we can see that for the same number of users and memory ratio, the schemes with transmission load for each relay from small to large are ZY scheme, Scheme B, Scheme A and Scheme C in turn. That is because ZY scheme fully utilizes the multicast opportunities among the users who are connected to the same relay, then its coded caching gain is the largest, so it has the minimum transmission load for each relay; Scheme A and Scheme C only utilizes the multicast opportunities among the users who are connected to different sets of $r$ relays, while omits the multicast opportunities among the users who are connected to the same $r$ relays, and Scheme A utilizes the multicast opportunities among the users who are connected to different sets of $r$ relays more efficiently than Scheme C, so the transmission load for each relay of Scheme A is lower than that of Scheme C; Scheme B not only utilizes the multicast opportunities among the users who are connected to different sets of $r$ relays as Scheme A, but also utilizes the multicast opportunities among the users who are connected to the same $r$ relays, so the transmission load for each relay of Scheme B is lower than that of Scheme A.
From Fig.\ref{fig:numerical}(b), we can see that for the same number of users and memory ratio, the schemes with subpacketization from small to large are Scheme A, Scheme C, Scheme B and ZY scheme in turn. So Scheme A and Scheme B have significant advantages in subpacketization compared with ZY scheme and have significant advantages in transmission load for each relay compared with Scheme C. Moreover, the subpacketization of Scheme A is lower than that of Scheme C.
\begin{figure}[http!]
    \centering
    \begin{subfigure}[t]{0.5\textwidth}
        \centering
        \includegraphics[scale=0.38]{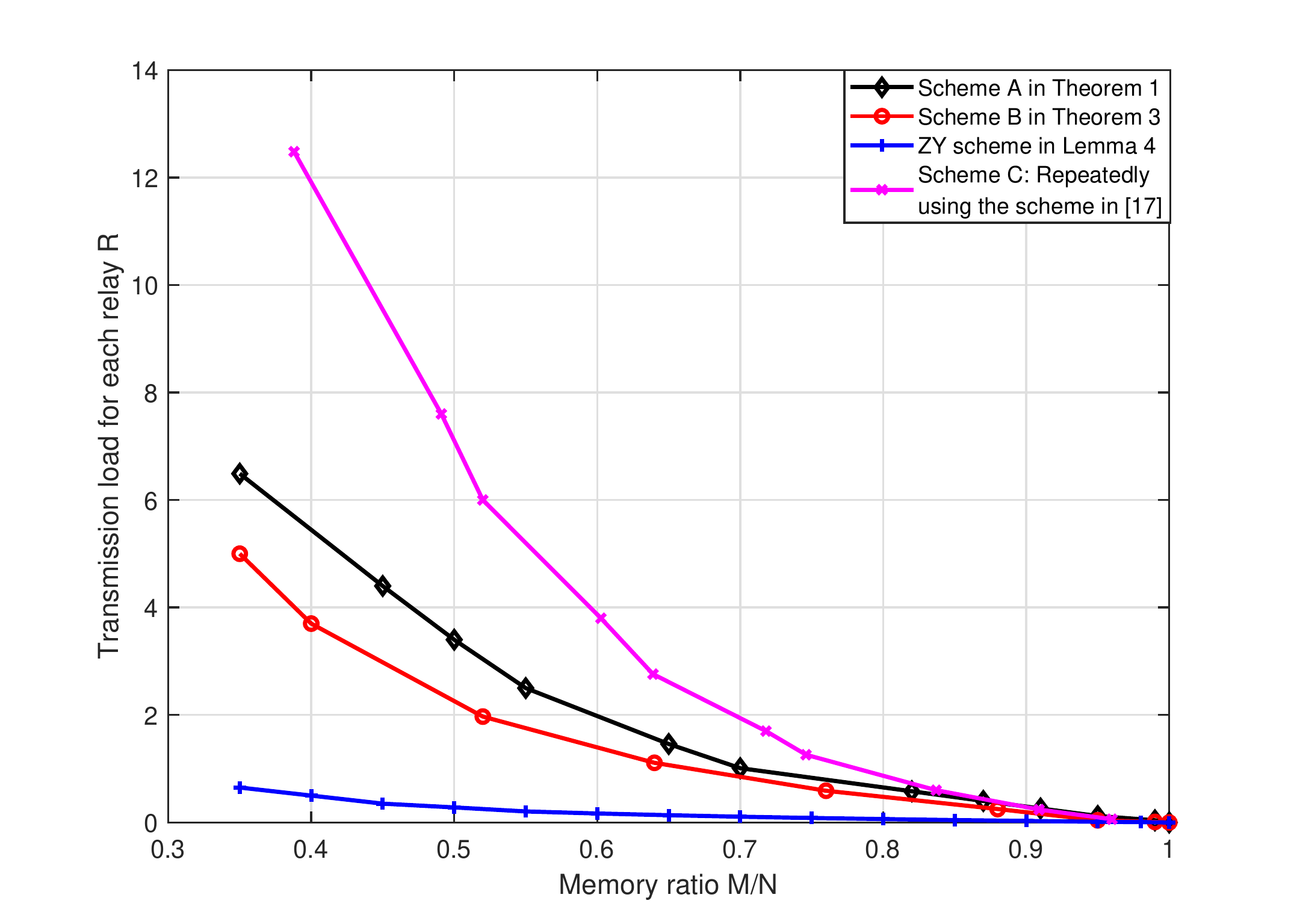}
        \caption{\small  Transmission load for each relay.}
        \label{a}
    \end{subfigure}%
    ~
    \begin{subfigure}[t]{0.5\textwidth}
        \centering
        \includegraphics[scale=0.38]{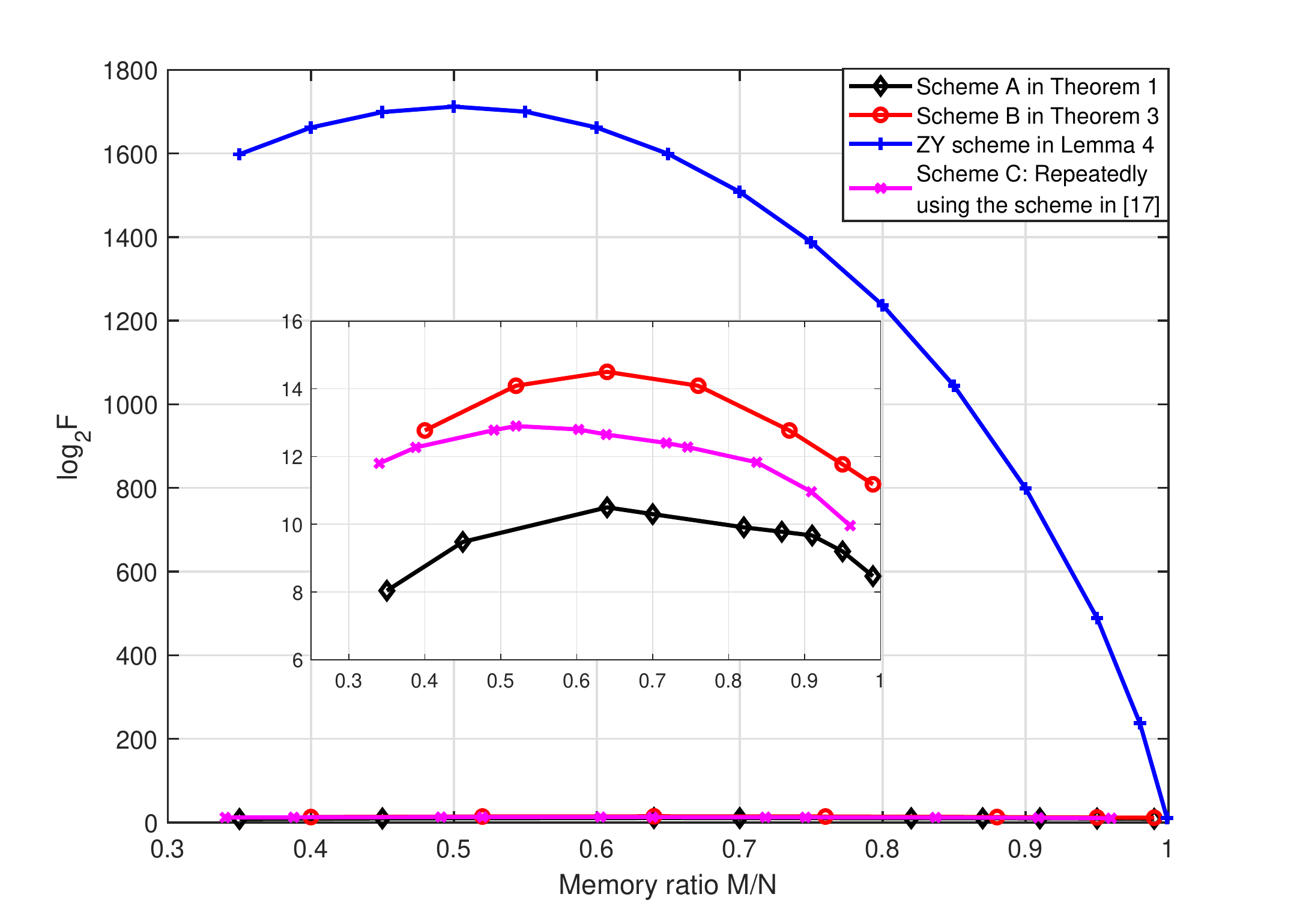}
        \caption{\small Subpacketization. }
        \label{b}
    \end{subfigure}
    \caption{\small Comparison of Scheme A and Scheme B with ZY scheme and Scheme C when $H=14$, $r=4$, $u=6$.}
    \label{fig:numerical}
\end{figure}

\section{conclusion}
\label{sec-conclusion}

In this paper, we extended a traditional $(H, r)$ combination network to an $(H,r,u)$ multiaccess combination network, where each $u$ users are connected to a unique set of $r$ relays, and proposed three schemes for such a network, i.e., ZY scheme, Scheme A and Scheme B. The transmission load for each relay of ZY scheme obtained from the idea in \cite{ZY2} is approximately $\frac{1}{r}$ of the transmission load of MN scheme, but its subpacketization increases exponentially with the number of users. Scheme A obtained by a direct construction of CPDA has significant advantages in subpacketization compared with ZY scheme, where the parameter $u$ must be a combinational number. Scheme B obtained by a hybrid construction of CPDA works for arbitrary parameter $u$ and also has significant advantages in subpacketization compared with ZY scheme. In addition, Scheme A and Scheme B have significant advantages in transmission load for each relay compared with Scheme C of repeatedly using the scheme in \cite{ZCWZ}. Moreover, the subpacketization of Scheme A is lower than that of Scheme C. It is worth noting that the hybrid construction is based on a given CPDA and PDA, and the memory ratio of the resulting CPDA is greater than that of the original CPDA. It is great significance to look for a direct construction of CPDA which works for arbitrary parameter $u$.

\appendices
\section{Proof of Theorem~\ref{th-2}}
\label{appendix-th-2}

\begin{proof}Let the array generated by Construction \ref{constr-1} be $\mathbf{P}=(p_{\mathbf{f},\mathbf{k}})$ where $\mathbf{f}\in \mathcal{F}$ and $\mathbf{k}\in \mathcal{K}$. For any non-star entry $(\mathbf{e},\mathcal{C})$ occurring in $\mathbf{P}$, i.e., there exists some $\mathbf{f}\in \mathcal{F}$ and $\mathbf{k}=(\mathcal{T} , \mathbf{b})\in \mathcal{K}$ such that $p_{\mathbf{f},(\mathcal{T}, \mathbf{b})}=(\mathbf{e},\mathcal{C})$. Let $\mathcal{C}=\mathcal{C}_1\bigcup \mathcal{C}_2$ where $\mathcal{C}_1$ is the subset of $\mathcal{C}$ such that $f_i=0$ for each $i\in \mathcal{C}_1$ and $\mathcal{C}_2$ is the subset of $\mathcal{C}$ such that $f_i=1$ for each $i\in \mathcal{C}_2$. Let $|\mathcal{C}_1|=\lambda_1$, then $|\mathcal{C}_2|=\lambda-\lambda_1$ since $\mathcal{C}_1\bigcap\mathcal{C}_2=\emptyset$. In addition, $0\leq \lambda_1\leq \omega$ and $0\leq \lambda-\lambda_1\leq r-\omega$. This implies that $\max\{0,\lambda-r+\omega\}\leq \lambda_1\leq \min\{\lambda,\omega\}$.
Let $x=\max\{0,\lambda-r+\omega\}$ and $y=\min\{\lambda,\omega\}$, then $\lambda_1\in[x:y]$.

Next we will prove that the array $\mathbf{P}$ satisfies the conditions of Definition \ref{def-CPDA}. Let us consider C1 of  Definition \ref{def-PDA} first. For any column label ${\bf k}=(\mathcal{T},{\bf b})$, without loss of generality, assume that
${\bf b}=(b_1,b_2,\ldots,b_{r})$ satisfying $b_1=b_2=\ldots=b_{\omega}=0$ and $b_{\omega+1}=\ldots=b_{r}=1$.
Let $\mathcal{T}=\{ \delta_1,\delta_2,\ldots, \delta_{r}\}$ with $\delta_1<\delta_2<\ldots< \delta_{r}$. From \eqref{eq-defin-array} there are
\begin{eqnarray}
\label{eq-number of nonstars}
\overline{Z}=\sum\limits_{\lambda_1\in [x:y]}\underbrace{{\omega\choose \lambda_1}}_{X_1}\times \underbrace{{r-\omega\choose \lambda-\lambda_1}}_{X_2}\times\underbrace{{H-r\choose a-\omega+2\lambda_1-\lambda}}_{X_3}
\end{eqnarray}
vectors ${\bf f}=(f_1,f_2,\ldots,f_{H})\in \mathcal{F}$ satisfying $\text{d}({\bf f}|_{\mathcal{T}},{\bf b})=r-\lambda$, where
 $X_1$ is the number of subvectors ${\bf f}|_{\{ \delta_1,\ldots, \delta_{\omega}\}}$ satisfying $\text{d}({\bf f}|_{\{ \delta_1,\ldots, \delta_{\omega}\}},{\bf b}|_{[1:\omega]})=\omega-\lambda_1$;
 $X_2$ is the number of subvectors ${\bf f}|_{\{ \delta_{\omega+1},\ldots, \delta_{r}\}}$ satisfying $\text{d}({\bf f}|_{\{ \delta_{\omega+1},\ldots, \delta_{r}\}},{\bf b}|_{[\omega+1:r]})=r-\omega-(\lambda-\lambda_1)$; and
 $X_3$ is the number of subvectors ${\bf f}|_{[H]\setminus\mathcal{T}}$ satisfying $\text{wt}({\bf f}|_{[H]\setminus\mathcal{T}})=a-(\omega-\lambda_1)-(\lambda-\lambda_1)=a-\omega+2\lambda_1-\lambda$.
Thus, there are $Z={H\choose a}-\overline{Z}$ stars in each column. So the condition C1 holds.

For any two distinct entries, say $p_{\mathbf{f},(\mathcal{T} , \mathbf{b})}$ and $p_{\mathbf{f}^\prime,(\mathcal{T}^\prime, \mathbf{b}^\prime)}$, assume that $p_{\mathbf{f},(\mathcal{T}, \mathbf{b})}=p_{\mathbf{f}^\prime,(\mathcal{T}^\prime , \mathbf{b}^\prime)}=(\mathbf{e},\mathcal{C})$. Denote
\begin{eqnarray*}
\begin{array}{ccc}
\mathbf{f}=(f_1, f_2, \ldots, f_{H}), & \mathcal{T}=\{\delta_1, \delta_2, \ldots, \delta_{r}\}, &\mathbf{b}=(b_1,b_2, \ldots,b_{r}), \\
\mathbf{f}^\prime= (f^\prime_1, f^\prime_2, \ldots, f^\prime_{H}), &\mathcal{T}^\prime =\{ \delta^\prime_1, \delta^\prime_2,\ldots, \delta^\prime_{r}\}, & \mathbf{b}^\prime = (b^\prime_1,b^\prime_2, \ldots,b^\prime_{r}).
\end{array}
\end{eqnarray*}
If $\mathcal{T} = \mathcal{T}^\prime$, we have ${\bf e}|_{\mathcal{T}}={\bf b}={\bf b}^\prime$ from \eqref{eq-defin-array}. Then we have
$\mathbf{f}=\mathbf{f}^\prime$ since
\begin{eqnarray*}
{\bf e}|_{[H]\setminus\mathcal{T}}={\bf f}|_{[H]\setminus\mathcal{T}}={\bf f}'|_{[H]\setminus\mathcal{T}}, \ \
{\bf e}|_{\mathcal{C}}={\bf f}|_{\mathcal{C}}={\bf f}'|_{\mathcal{C}},\ \
{\bf f}|_{\mathcal{T}\setminus\mathcal{C}}\neq {\bf e}|_{\mathcal{T}\setminus\mathcal{C}},\ \
{\bf f}'|_{\mathcal{T}\setminus\mathcal{C}}\neq {\bf e}|_{\mathcal{T}\setminus\mathcal{C}}.
\end{eqnarray*} This contradicts our hypothesis\footnote{\label{foot:1} It implies that for any $r$-relay set $\mathcal{T}\in {[H]\choose r}$, each non-star entry $(\mathbf{e},\mathcal{C})$ in $\mathbf{P}$ occurs at most once in the columns $\mathbf{k}=(\mathcal{T},\mathbf{b})$ of $\mathbf{P}$ where ${\bf b}\in \{0,1\}^r$ and $\text{wt}({\bf b})=r-\omega$. So the multicast opportunities among the users who are connected to the same $r$ relays are not utilized.}. So we have $\mathcal{T}\neq \mathcal{T}^\prime$, then there must exist two distinct integers, say $i$, $i'\in[H]$, satisfying
$i\in \mathcal{T}\setminus\mathcal{C},\  i\not\in\mathcal{T}'$ and $i'\in \mathcal{T}'\setminus\mathcal{C},\  i'\not\in\mathcal{T}$.
Without loss of generality, assume that $i=\delta_1$ and $i'=\delta'_1$. From Construction \ref{constr-1}, we have
$f_{\delta_1}\neq b_{1}=e_{\delta_1}=f'_{\delta_1}$ and $f'_{\delta'_1}\neq b'_1 =e_{\delta'_1}=f_{\delta'_1}$.
From \eqref{eq-defin-array}, \eqref{eq-entry} and \eqref{eq-c} we have $\mathcal{C}\subseteq\mathcal{T}\bigcap \mathcal{T}'$ and ${\bf f}|_{\mathcal{C}}={\bf b}|_{\mathcal{C}}={\bf e}|_{\mathcal{C}}={\bf f'}|_{\mathcal{C}}={\bf b'}|_{\mathcal{C}}$. Then $d({\bf f}|_{\mathcal{T}'},{\bf b}')<r-\lambda$ and $d({\bf f}'|_{\mathcal{T}},{\bf b})<r-\lambda$ hold since $|\mathcal{C}|=\lambda$. So we have $p_{{\bf f}, (\mathcal{T}',{\bf b}')}=p_{{\bf f}', (\mathcal{T},{\bf b})}=*$ from \eqref{eq-defin-array}. Clearly ${\bf f}\neq {\bf f}'$ holds. So the condition C$3$ of Definition \ref{def-PDA} holds.

Furthermore, it implies that each useful star $p_{{\bf f}, (\mathcal{T},{\bf b})}=*$  satisfies $d({\bf f}|_{\mathcal{T}},{\bf b})<r-\lambda$. In other words, each star $p_{{\bf f}, (\mathcal{T},{\bf b})}=*$ satisfying $d({\bf f}|_{\mathcal{T}},{\bf b})>r-\lambda$ is useless. Similar to \eqref{eq-number of nonstars}, there are exactly
$Z'=\sum\limits_{\lambda'\in[0:\lambda-1]}\sum\limits_{\lambda_1\in[x':y']}{\omega\choose \lambda_1}\times {r-\omega\choose \lambda'-\lambda_1}\times{H-r\choose a-\omega+2\lambda_1-\lambda'}$ vectors ${\bf f}\in \mathcal{F}$ satisfying $\text{d}({\bf f}|_{\mathcal{T}},{\bf b})>r-\lambda$, where $x'=\max\{0,\lambda'-r+\omega \}$ and $y'=\min\{\lambda',\omega\}$. Deleting the useless stars, there are exactly $Z-Z'$ stars in each column.

For any non-star entry $(\mathbf{e},\mathcal{C})$ occurring in $\mathbf{P}$, assume that $p_{\mathbf{f},(\mathcal{T}, \mathbf{b})}=(\mathbf{e},\mathcal{C})$, from Construction \ref{constr-1} we have $\mathcal{C}\subset\mathcal{T}$, $|\mathcal{C}|=\lambda$, $e_i=f_i$ for each $i\in \mathcal{C}\cup ([H]\setminus \mathcal{T})$ and $e_i\neq f_i$ for each $i\in \mathcal{T} \setminus \mathcal{C}$. Since $\mathcal{C}=\mathcal{C}_1\bigcup \mathcal{C}_2$ where $\mathcal{C}_1$ is the subset of $\mathcal{C}$ such that $f_i=0$ for each $i\in \mathcal{C}_1$, $|\mathcal{C}_1|=\lambda_1\in[x:y]$ and $\mathcal{C}_2$ is the subset of $\mathcal{C}$ such that $f_i=1$ for each $i\in \mathcal{C}_2$, we have $wt({\bf f}|_{\mathcal{C}})=\lambda-\lambda_1$. Then $wt({\bf f}|_{\mathcal{T}\setminus \mathcal{C}})=\omega-\lambda_1$ since $wt({\bf b})=r-\omega$ and $d({\bf f}|_{\mathcal{T}},{\bf b})=r-\lambda$. So $wt({\bf f}|_{\mathcal{T}})=\lambda-\lambda_1+\omega-\lambda_1$ and $wt({\bf f}|_{[H]\setminus \mathcal{T}})=a-(\lambda+\omega-2\lambda_1)$ since $wt({\bf f})=a$. Since $wt({\bf e}|_{\mathcal{T}})=wt({\bf b})=r-\omega$ and $wt({\bf e}|_{[H]\setminus\mathcal{T}})=wt({\bf f}|_{[H]\setminus\mathcal{T}})=a-(\lambda+\omega-2\lambda_1)$, we have $wt({\bf e})=r-\omega+a-(\lambda+\omega-2\lambda_1)=a+r-2\omega+2\lambda_1-\lambda$.

    For any $\lambda_1\in [x:y]$ and for each vector ${\bf e}$ with wt$({\bf e})=a+r-2\omega+2\lambda_1-\lambda$, let us consider any possible coordinate sets $\mathcal{C}_1$ and $\mathcal{C}_2$ such that 1) each coordinate of ${\bf e}|_{\mathcal{C}_1}$ is zero; 2) each coordinate of ${\bf e}|_{\mathcal{C}_2}$ is one; 3) $|\mathcal{C}_1|=\lambda_1$ and $|\mathcal{C}_2|=\lambda-\lambda_1$. Let $\mathcal{C}=\mathcal{C}_1\bigcup\mathcal{C}_2$. Next we will prove that $({\bf e},\mathcal{C})$ occurs in $\mathbf{P}$. Without loss of generality, assume that $\mathcal{C}_1=[1:\lambda_1]$, $\mathcal{C}_2=[\lambda_1+1:\lambda]$ and
$e_1=e_2=\cdots=e_{\lambda_1}=0$, $e_{\lambda_1+1}=e_{\lambda_1+2}=\cdots=e_{\lambda}=\cdots=e_{a+r-2\omega+3\lambda_1-\lambda }=1$ and $e_{a+r-2\omega+3\lambda_1-\lambda+1}=e_{a+r-2\omega+3\lambda_1-\lambda+2}=\cdots=e_{H}=0$.
That is,
$$\mathbf{e}=(\underbrace{0,\ \cdots,\  0}_{\lambda_1},\underbrace{1,\ \cdots,\ 1}_{\lambda-\lambda_1},\underbrace{1,\ \ \ \ \cdots,\ \ \ \ 1,}_{a+r-2\omega+3\lambda_1-2\lambda}\ \ \underbrace{0,\ \ \ \ \ \ \cdots,\ \ \ \ \ \  0)}_{H-(a+r-2\omega+3\lambda_1-\lambda)}.$$
Define
\begin{small}
\begin{eqnarray*}
\mathfrak{T}_{\mathcal{C}}=\left\{\mathcal{C}\cup\mathcal{T}_1\cup\mathcal{T}_2 \ \Big|\
\mathcal{T}_1\in {[a+r-2\omega+3\lambda_1-\lambda+1:H]\choose \omega-\lambda_1},\mathcal{T}_2\in{[\lambda+1:a+r-2\omega+3\lambda_1-\lambda]\choose r-\omega-(\lambda-\lambda_1)}\right\}.
\end{eqnarray*}
\end{small} It is easy to check that for any $\mathcal{T}=\mathcal{C}\cup\mathcal{T}_1\cup\mathcal{T}_2\in\mathfrak{T}_{\mathcal{C}}$, we have $|\mathcal{T}|=r$, $wt({\bf e}|_{\mathcal{T}})=r-\omega$ and there is an $H$-dimensional vector {\bf f} defined as 
\begin{eqnarray*}
f_{i}=\left\{\begin{array}{cc}
        1 & \text{if}\ \ i\in \mathcal{T}_1, \\
        0 & \text{if}\ \ i\in \mathcal{T}_2,\\
        e_i & \text{otherwise}
      \end{array}\right.
\end{eqnarray*}
 satisfying $\text{wt}({\bf f})=a$, then ${\bf f}\in\mathcal{F}$, $(\mathcal{T},{\bf e}|_{\mathcal{T}})\in\mathcal{K}$ and $d({\bf f}|_{\mathcal{T}},{\bf e}|_{\mathcal{T}})=r-\lambda$. Consequently, we have $p_{{\bf f},(\mathcal{T},{\bf e}|_{\mathcal{T}})}=({\bf e}, \mathcal{C})$ by Construction \ref{constr-1}. And $({\bf e}, \mathcal{C})$ appears ${H-(a+r-2\omega+3\lambda_1-\lambda)\choose \omega-\lambda_1}{a+r-2\omega+3\lambda_1-2\lambda\choose r-\omega-(\lambda-\lambda_1)}$ times in $\mathbf{P}$. Since the intersection of any element of ${[a+r-2\omega+3\lambda_1-\lambda+1:H]\choose \omega-\lambda_1}$ and any element of ${[\lambda+1:a+r-2\omega+3\lambda_1-\lambda]\choose r-\omega-(\lambda-\lambda_1)}$ is empty, we have $\bigcap_{\mathcal{T}\in \mathfrak{T}_{\mathcal{C}}}\mathcal{T}=\mathcal{C}$. This implies that the set $\mathcal{I}_{({\bf e}, \mathcal{C})}=\mathcal{C}$. Recall that $\mathcal{I}_{({\bf e}, \mathcal{C})}$ is the intersection of the first coordinate $\mathcal{T}$ of all column labels $(\mathcal{T}, {\bf b})$ satisfying that $({\bf e}, \mathcal{C})$ appears in column $(\mathcal{T}, {\bf b})$, i.e., $\mathcal{I}_{({\bf e}, \mathcal{C})}=\cap_{p_{{\bf f},(\mathcal{T}, {\bf b})}=({\bf e}, \mathcal{C}),{\bf f}\in\mathcal{F},(\mathcal{T}, {\bf b})\in\mathcal{K}}\mathcal{A}_{(\mathcal{T}, {\bf b})}$, where $\mathcal{A}_{(\mathcal{T}, {\bf b})}$ is the set of relays which connect user $(\mathcal{T}, {\bf b})$. The condition C$4$ of Definition \ref{def-CPDA} holds.

From the above analysis, there are $S=\sum\limits_{\lambda_1\in[x:y]}{H\choose a+r-2\omega+2\lambda_1-\lambda}{H-(a+r-2\omega+2\lambda_1-\lambda)\choose \lambda_1}{a+r-2\omega+2\lambda_1-\lambda\choose \lambda- \lambda_1}$ different non-star entries in $\mathbf{P}$. The condition C2 of Definition \ref{def-PDA} holds.

Hence, $\mathbf{P}$ is a CPDA with the parameters in Theorem \ref{th-2}. For each non-star entry $({\bf e}, \mathcal{C})$ in $\mathbf{P}$, since $\mathcal{I}_{({\bf e}, \mathcal{C})}=\mathcal{C}$, we have $|\mathcal{I}_{({\bf e}, \mathcal{C})}|=\lambda$, which is a constant. Next we will prove that for each relay $h\in[H]$, the number of $\mathcal{I}_{({\bf e}, \mathcal{C})}$ containing $h$ is also a constant. In fact, the number of $\mathcal{I}_{({\bf e}, \mathcal{C})}$ containing relay $h$ is exactly
\begin{eqnarray*}
&\sum\limits_{\lambda_1\in[x:y]}{H-1\choose a+r-2\omega+2\lambda_1-\lambda}
{H-(a+r-2\omega+2\lambda_1-\lambda)-1\choose \lambda_1-1}
{a+r-2\omega+2\lambda_1-\lambda\choose \lambda- \lambda_1}+\\
&\sum\limits_{\lambda_1\in[x:y]}{H-1\choose a+r-2\omega+2\lambda_1-\lambda-1}
{H-(a+r-2\omega+2\lambda_1-\lambda)\choose \lambda_1}
{a+r-2\omega+2\lambda_1-\lambda-1\choose \lambda- \lambda_1-1}\\
=&\frac{\lambda}{H}\sum\limits_{\lambda_1\in[x:y]}{H\choose a+r-2\omega+2\lambda_1-\lambda}
{H-(a+r-2\omega+2\lambda_1-\lambda)\choose \lambda_1}
{a+r-2\omega+2\lambda_1-\lambda\choose \lambda- \lambda_1},
\end{eqnarray*}
which is also a constant.
From Lemma \ref{le-coded CPDA}, Theorem \ref{th-2} is proved.

\end{proof}

\section{Proof of Theorem~\ref{th-3}}
\label{appendix-th-3}
\begin{proof}
Firstly, we will prove that the array $\mathbf{L}$ generated by Construction \ref{constr-2} is a $(uK_2{H\choose r}, F_1F_2$, $Z_1F_2+(F_1-Z_1)Z_2$, $S_1S_2)$ CPDA. Clearly, $\mathbf{L}$ is a $uK_2{H\choose r}\times F_1F_2$ array.
Furthermore, each column of $\mathbf{L}$ has $Z_1F_2+(F_1-Z_1)Z_2$ stars since each column of $\mathbf{P}$ has $Z_1$ stars and each column of $\mathbf{A}$ has $Z_2$ stars. There are $S=S_1S_2$ integers in $\mathbf{L}$, since there are $S_1$ integers in $\mathbf{P}$ and there are $S_2$ integers in $\mathbf{A}$. The Conditions C$1$ and C$2$ of Definition \ref{def-PDA} hold.

For any two distinct entries $l_{(j_1,j_2),({\bf k_1},k_2)}$ and $l_{(j'_1,j'_2),({\bf k'_1},k'_2)}$, if $l_{(j_1,j_2),({\bf k}_1,k_2)}=l_{(j'_1,j'_2),({\bf k}'_1,k'_2)}=s\in[S]$, then all entries $a_{j_2,k_2},a_{j'_2,k'_2},p_{j_1,{\bf k}_1},p_{j'_1,{\bf k}'_1}$ are integers and $s=a_{j_2,k_2}+(p_{j_1,{\bf k}_1}-1)S_2=a_{j'_2,k'_2}+(p_{j'_1,{\bf k}'_1}-1)S_2$ from \eqref{const-2}. Since $a_{j_2,k_2}\leq S_2$ and $a_{j'_2,k'_2}\leq S_2$, we have $a_{j_2,k_2}=a_{j'_2,k'_2}\in[S_2]$ and $p_{j_1,{\bf k}_1}=p_{j'_1,{\bf k}'_1}\in[S_1]$. Then from the definition of PDA and CPDA, we have that (1) either $j_2=j'_2, k_2=k'_2$ or $a_{j_2,k'_2}=a_{j'_2,k_2}=*$ holds; (2) either $j_1=j'_1, {\bf k}_1={\bf k}'_1$ or $p_{j_1,{\bf k}'_1}=p_{j'_1,{\bf k}_1}=*$ holds. Since $l_{(j_1,j_2),({\bf k_1},k_2)}$ and $l_{(j'_1,j'_2),({\bf k'_1},k'_2)}$ are two distinct entries, $j_2=j'_2, k_2=k'_2$ and $j_1=j'_1, {\bf k}_1={\bf k}'_1$ can not hold simultaneously, which implies that either $a_{j_2,k'_2}=a_{j'_2,k_2}=*$ or $p_{j_1,{\bf k}'_1}=p_{j'_1,{\bf k}_1}=*$ holds. Consequently, we have $l_{(j_1,j_2),({\bf k}'_1,k'_2)}=l_{(j'_1,j'_2),({\bf k}_1,k_2)}=*$ from \eqref{const-2}. The condition C$3$ of Definition \ref{def-PDA} holds.

For any $s\in[S]$, assume that $s$ appears $g_s$ times in $\mathbf{L}$, say $l_{(j_{1}^{(1)},j_2^{(1)}),({\bf k}_1^{(1)},k_2^{(1)})}=l_{(j_{1}^{(2)},j_2^{(2)}),({\bf k}_1^{(2)},k_2^{(2)})}=\ldots=l_{(j_{1}^{(g_s)},j_2^{(g_s)}),({\bf k}_1^{(g_s)},k_2^{(g_s)})}=s$, then $p_{j_{1}^{(1)},{\bf k}_1^{(1)}}=p_{j_{1}^{(2)},{\bf k}_1^{(2)}}=\ldots=p_{j_{1}^{(g_s)},{\bf k}_1^{(g_s)}}=\frac{s-<s>_{S_2}}{S_2}+1\triangleq s_1$ from \eqref{const-2}. Assume that ${\bf k}_1^{(\mu)}=(\mathcal{T}_{\mu},i_{\mu})$ for each $\mu\in[g_s]$, we have $\mathcal{I}_{s}^{(\mathbf{L})}=\mathcal{I}_{s_1}^{(\mathbf{P})}=\cap_{\mu\in[g_s]}\mathcal{T}_{\mu}$.
Since $\mathbf{P}$ is a CPDA, from condition C$4$ of Definition \ref{def-CPDA}, the set $\mathcal{I}_{s_1}^{(\mathbf{P})}$ is not empty. Then $\mathcal{I}_{s}^{(\mathbf{L})}$ is not empty and the condition C$4$ of Definition \ref{def-CPDA} holds. Hence, the array $\mathbf{L}$ is a $(K,F,Z,S)=(uK_2{H\choose r}, F_1F_2,Z_1F_2+(F_1-Z_1)Z_2,S_1S_2)$ CPDA.

Secondly, if $|\mathcal{I}_{s_1}^{(\mathbf{P})}|=\mu$ is a constant for each $s_1\in[S_1]$ and the number of $\mathcal{I}_{s_1}^{(\mathbf{P})}$ containing $h$ is a constant for each $h\in[H]$, say $|\{\mathcal{I}_{s_1}^{(\mathbf{P})}|h\in\mathcal{I}_{s_1}^{(\mathbf{P})},s_1\in[S_1]\}|=\nu$ for each $h\in[H]$, then we have $|\mathcal{I}_s^{(\mathbf{L})}|=|\mathcal{I}_{s_1}^{(\mathbf{P})}|=\mu$ for each $s\in[S_1S_2]$ where $s_1=\frac{s-<s>_{S_2}}{S_2}+1$ and $|\{\mathcal{I}_s^{(\mathbf{L})}|h\in\mathcal{I}_s^{(\mathbf{L})},s\in[S_1S_2]\}|=\nu S_2$ for each $h\in[H]$. If there are $Z'_1$ useless stars in each column of $\mathbf{P}$, there are $Z'_1F_2$ useless stars in each column of $\mathbf{L}$.
From Lemma \ref{le-coded CPDA}, $\mathbf{L}$ leads to an $(H,r,uK_2,M,N)$ multiaccess combination network caching scheme with memory ratio $\frac{M}{N}=\frac{Z_1F_2+(F_1-Z_1)Z_2-Z'_1F_2}{F_1F_2-Z'_1F_2}=\frac{Z_1-Z'_1}{F_1-Z'_1}+\left(1-\frac{Z_1-Z'_1}{F_1-Z'_1}\right)\frac{Z_2}{F_2}$, subpacketization $F=\mu(F_1-Z'_1)F_2$ and transmission load for each relay $R=\frac{S_1S_2}{H(F_1-Z'_1)F_2}$.
\end{proof}
\bibliographystyle{IEEEtran}
\bibliography{reference}

\end{document}